\documentclass[journal]{IEEEtran}

\IEEEoverridecommandlockouts
\usepackage{amsmath}
\usepackage{amsfonts}
\usepackage{amssymb}
\usepackage{graphicx}
\usepackage{subcaption}
\usepackage{epstopdf}
\usepackage{cite}
\usepackage{algorithmicx}
\usepackage[ruled]{algorithm} 
\usepackage{algpseudocode}
\usepackage{blindtext}
\usepackage{enumitem}
\usepackage{pdfpages}
\usepackage{bbm}
\usepackage{bibentry}
\usepackage{amsthm}
\usepackage{color}
\usepackage{mathtools}
\newtheorem{theorem}{Theorem}
\newtheorem{lemma}[theorem]{Lemma}
\usepackage[normalem]{ulem}

\begin{document}
\bstctlcite{IEEEexample:BSTcontrol}
%
\title{Digital and Hybrid Precoding and RF Chain Selection Designs for Energy Efficient Multi-User MIMO-OFDM ISAC Systems}

\author{Po-Chun Kang, Ming-Chun Lee,~\IEEEmembership{Member,~IEEE},  Tzu-Chien Chiu, Ting-Yao Kuo and Ta-Sung Lee, \IEEEmembership{Fellow, IEEE}
\thanks{P.-C. Kang, M.-C. Lee, T.-Z. Chiu, T.-Y. Kuo and T.-S. Lee are with Institute of Communications Engineering, National Yang Ming Chiao Tung University, Hsinchu 30010, Taiwan. (email: kangpopo2255.ee11@nycu.edu.tw, mingchunlee@nycu.edu.tw, howard891108@gmail.com, guo.ting.Yao.ee13@nycu.edu.tw, tslee@nycu.edu.tw)}
}

\maketitle

\begin{abstract}
Using multiple-input multiple-output (MIMO) with orthogonal frequency division multiplexing (OFDM) for integrated sensing and communication (ISAC) has attracted considerable attention in recent years. While most existing works focus on improving MIMO-OFDM ISAC performance, the impact of transmit power and radio-frequency (RF) circuit power consumption on energy efficiency (EE) remains relatively underexplored. To address this gap, this paper investigates joint precoding and RF chain selection for multi-user MIMO-OFDM ISAC systems, and develops energy-efficient designs for both fully digital and hybrid precoding architectures through the joint optimization of precoding and RF-chain activation. Specifically, we first formulate a novel EE maximization problem subject to sensing performance constraints. Then, efficient optimization algorithms are proposed for both architectures, together with analyses of their computational complexity and convergence behavior. Building on the proposed approaches, spectral efficiency–power consumption tradeoff designs are also provided. Simulation results demonstrate that, compared with existing schemes, the proposed approaches achieve significant improvements in the EE–sensing tradeoff for ISAC systems.
\end{abstract}

\vspace{-3pt}
\IEEEpeerreviewmaketitle
\vspace{-10pt}

\section{Introduction}

To enable emerging applications and alleviate spectrum competition between communication and sensing systems, the concept of integrated sensing and communication (ISAC) has been extensively investigated in both academia and industry \cite{ZhangJCS2022,Lu2024ISAC10}. Among various ISAC architectures, communication-centric systems that employ conventional orthogonal frequency division multiplexing (OFDM) waveforms to support both radar sensing and data transmission have attracted considerable attention \cite{braun2014ofdm,carvajal2020comparison,dai2025tutorial}. Meanwhile, the increasing complexity of wireless communication systems and the growing performance demands of modern applications have led to a significant rise in system energy consumption, making energy-efficient system design a critical challenge \cite{Zhang20196G}. From a physical-layer perspective, radio-frequency (RF) circuit power often constitutes a substantial portion of the total energy consumption. Therefore, improving energy efficiency (EE) requires not only appropriate transmit power control but also careful management of circuit power, for example, by deactivating unnecessary RF chains \cite{Gao2016EE}. This observation has motivated a growing body of work that investigates the EE of ISAC systems \cite{he2022energy, ZouEEISAC2022, zou2024energy, allu2024robust, singh2024energy,kaushik2021hardware,kaushik2022green, dizdar2022energy,WuEEISAC2024}.

\subsection{Literature Review}

To address the EE issue in ISAC systems, \cite{he2022energy,ZouEEISAC2022,zou2024energy} took the transmit power consumption into account and designed beamforming focusing on improving EE. Their results showed that EE-oriented designs could be fairly different from the throughput- and sensing-oriented designs. Later, similar investigations in \cite{allu2024robust, singh2024energy} which additionally consider channel uncertainty and hybrid beamforming architecture also showed that EE can be improved by using EE-based beamforming designs, as compared to designs without regard to EE. 

When the number of antennas increases, the energy consumption of RF circuits becomes unignorable. Hence, works have also explored the joint design of precoding and RF chain selection, which further exploits the ON/OFF switching of RF chains to enhance EE \cite{kaushik2021hardware, kaushik2022green, dizdar2022energy}. Specifically, in \cite{kaushik2021hardware}, a joint design of hybrid precoder and RF chain selection was proposed which showed that the best EE can only be achieved by having minimum hardware cost. Then, in \cite{kaushik2022green}, the investigation in \cite{kaushik2021hardware} was extended to considering the resolution of the digital-to-analog converters (DACs) when jointly designing with the RF chain selection and hybrid precoder. Similarly, in \cite{dizdar2022energy}, considering the rate-splitting multiple access (RSMA) and low-resolution DACs in RF chains to reduce the power consumption of each RF chain, the relevant energy efficient joint precoding and RF chain selection designs were proposed. Furthermore, in \cite{WuEEISAC2024}, when considering the Internet of Things (IoT) devices along with their EE-oriented designs, the ON-OFF mechanism of the complete baseband circuit was introduced to reduce the energy consumption. While \cite{he2022energy, ZouEEISAC2022, zou2024energy, allu2024robust, singh2024energy,kaushik2021hardware, dizdar2022energy, kaushik2022green,WuEEISAC2024} investigated EE designs considering both precoding and hardware energy consumptions with various approaches and architectures, all of them considered simple narrowband MIMO ISAC systems, leaving the wideband MIMO ISAC systems untreated from the EE aspect.

On the other hand, OFDM-based ISAC systems, serving as promising wideband ISAC systems, have been studied in the past several years. For example, the fundamental tradeoff between sensing and communication in OFDM systems was investigated in \cite{huang2023capacity}; the peak-to-average power ratio (PAPR) issue was studied in \cite{huang2022designing}; and \cite{hsu2021analysis} investigated the pilot power allocation and placement. To achieve better ISAC performance, MIMO-OFDM ISAC systems that employ antenna arrays have been extensively discussed. For example, different target parameter estimation approaches were discussed in \cite{xu2020joint, xiao2024novel}. In addition, beamforming and precoding designs were investigated. For example, a Kullback-Leibler divergence (KLD)-based approach was proposed in \cite{tian2021transmit} to joint design transmit and receive beamformers. 
To enhance the MIMO-OFDM ISAC systems under stringent sensing requirement, \cite{nguyen2023multiuser} proposed a subcarrier allocation strategy with the corresponding multi-user hybrid precoding. 
To enable spatial multiplexing for MIMO-OFDM ISAC systems, \cite{wei2024precoding} studied the precoding optimization by maximizing the sensing mutual information while ensuring a required communication signal-to-interference-plus-noise ratio (SINR). In \cite{Xiao2025SpaseISAC}, to leverage the sparsity of echo signals, a joint transmit beamforming and target parameter estimation approach was proposed. As the high PAPR is a headache for MIMO-OFDM ISAC systems, \cite{Hu2022low} proposed a low-PAPR waveform design. Since the MIMO-OFDM ISAC system designs are commonly based on non-unified views for communication and sensing, \cite{wei2023waveform} studied the design under a unified mutual information viewpoint.

While using abstract design criteria, e.g., SINR, mutual information, and Cram\'er-Rao bound (CRB), the design problems could be more tractable. These approaches might not directly correspond to the detection procedure. To resolve this issue, \cite{liao2024beamforming,chou2024robust,liao2025design} investigated the range-Doppler map, and then devised beamforming approaches for MIMO-OFDM ISAC systems. Similarly, as the impact of modulation symbols could be overlooked in approaches relying on the high-level criteria, symbol-level precoding techniques that can incorporate the impact of symbols were discussed in \cite{liu2021dual,Li2025ISACSidelobe,Wu2025LowISAC}. 
Recently, to benefit from the spare array architectures that have been long investigated in the area of array signal processing, MIMO-OFDM ISAC systems in consideration of sparse arrays have started to draw attentions and been discussed in \cite{liu2024joint} and \cite{Li2025Sparse}.

\subsection{Contributions}

While numerous studies on MIMO-OFDM ISAC systems have been published, to the best of our knowledge, energy efficient design remains largely unaddressed, with existing EE studies primarily focusing on narrowband MIMO ISAC systems. On the other hand, to significantly enhance the EE of ISAC systems, the adoption of joint precoding and RF chain selection strategies is necessary, as transmit power consumption and hardware cost should be jointly considered. However, extending existing narrowband approaches to wideband MIMO-OFDM ISAC systems is non-trivial, as different subcarriers may experience distinct channel fading and careful consideration is required to ensure reliable sensing performance when precoding across subcarriers. Thus, a significant gap remains in the current ISAC literature concerning energy efficient design, and this is an especially critical issue given the pivotal role that MIMO-OFDM ISAC systems are expected to play in next-generation wireless networks \cite{Luo2025ISACStandard}.

To fill the gap, we in this paper conduct a comprehensive investigation on the joint precoding and RF chain selection approaches for energy-efficient multi-user MIMO-OFDM ISAC systems, where fully digital (FD) precoding and hybrid precoding of fully-connected (FC) and partially-connected (PC) architectures are considered. To this end, we first investigate the EE model and sensing design criterion, and then formulate the FD precoding and RF chain selection problem aiming to maximize EE while guaranteeing sensing performance. Subsequently, as the design problem is a nonlinear mixed-integer problem which is notoriously difficult to solve, we propose a design framework that solves the problem by first approximating the energy consumption model using the hyperbolic tangent function, and then transforming the problem into a sequence of subproblems, whose solutions are obtainable, with an iterative manner. To compensate with the proposed framework which relies on the hyperbolic tangent function, the greedy-based and brute-force-based approaches without using the approximation are also proposed at the cost of higher complexity. Since our precoding design approach for EE maximization can be easily extended to obtain a spectral efficiency (SE)-power consumption tradeoff design, the relevant tradeoff design approach is provided.

Equipped with the design approach for the FD precoding, design approaches for hybrid precoding and RF chain selection under both FC and PC architectures are proposed. The fundamental ideas of devising approaches under FC and PC architectures are to first derive their FD precoding and RF chain selection counterparts, and then let the analog and digital precoding and RF chain selection designs to approach their FD counterparts under revised RF circuit energy consumption models. Convergence and complexity of our proposed design approaches are analyzed. Computer simulations are conducted to validate the efficacy of our approaches. Results show that our approaches can provide much better EE-sensing tradeoff than the reference schemes in the literature. Furthermore, results also validate that our proposed approaches can appropriately adjust the ON/OFF switching of RF chains. The main contributions of the paper are summarized as follows:
\begin{itemize}
    \item To fill the gap that joint precoding and RF chain selection for EE maximization in multi-user MIMO-OFDM ISAC systems has not been explored, we conduct a study on the design approaches for both FD precoding and hybrid precoding architectures. To the best of our knowledge, this paper is the first to conduct such investigation.
    \item By investigating the EE model and sensing criterion, joint precoding and RF chain selection design problems for EE maximization under MIMO-OFDM ISAC systems of FD precoding and hybrid precoding architectures are proposed. To solve the problems, novel solution approaches are proposed. Also, the extension to SE-power consumption tradeoff design is provided. 
    \item Theoretical convergence and complexity analysis of our proposed approaches are presented. In addition, computer simulations are conducted to evaluate our approaches. Results clearly indicate the superiority of our proposed design approaches as compared to reference schemes.
\end{itemize}

\section{System Model}
\label{sec: 2}

In this paper, we consider a base station (BS) equipped with a MIMO-OFDM ISAC transmitter that simultaneously serves $N_\text{UE}$ communication users and senses $N_\text{tar}$ targets. The BS is equipped with $N_\text{t}$ transmit antennas and $N_\text{r}^\text{sen}$ receive radar antennas for receiving the echoes from sensing targets. Each user is equipped with $N_\text{r}$ receive antennas to receive communication signals from the BS. We assume that the point-target assumption holds and consider that all targets to sense are line-of-sight (LoS) targets.

%
%

\subsection{Communication Signal Model}
\label{sec: 2-2}
We consider that the OFDM signal includes $N_\text{sub}$ subcarriers and $N_\text{sym}$ OFDM symbols, and the BS aims to send each user $N_s$ data streams in each subcarrier. Then, with the standard MIMO-OFDM processing procedure, the equivalent received baseband signal for user $u$ on subcarrier $k$ of OFDM symbol $l$ can be expressed as:
\begin{equation}
\label{eq: userRxSignal}
\begin{aligned}
    \mathbf{y}_{k,l}^{u} &= \mathbf{H}_{k,u}\left(\mathbf{F}_{k,u}\mathbf{s}_{k,l}^{u} + \sum_{i=1, i\neq u}^{N_\text{UE}}\mathbf{F}_{k,i}\mathbf{s}_{k,l}^{i}\right) + \mathbf{n}_{k,l}^{u},
\end{aligned}
\end{equation}
where 
$\mathbf{H}_{k,u}$, $\mathbf{F}_{k,u}$, and  $\mathbf{s}_{k,l}^{u}$ are the channel matrix, precoding matrix, and transmitted symbol for user $u$ on subcarrier $k$ and OFDM symbol $l$, respectively, and $\mathbf{n}_{k,l}^{u}$ is the white Gaussian noise with variance being $\sigma_\text{n}^2$. Note that here we consider that multiple OFDM symbols share the same precoder and channel matrices, reflecting the practice where the precoder typically changes only once across several symbols and the channel coherence time is larger than the symbol duration of the OFDM symbol. That being said, extending our proposed design approach below to other considerations is feasible and straightforward. We let $\mathbf{F}_k = \left[\mathbf{F}_k^{1},\ldots,\mathbf{F}_k^{N_\text{UE}}\right]$ and $\mathbf{F} = [\mathbf{F}_1,\ldots,\mathbf{F}_{N_\text{sub}}]$. Then, with (\ref{eq: userRxSignal}), the SE is given as: 
\begin{equation}
\label{eq: SE}
    R\left(\mathbf{F}\right)=\frac{\displaystyle{\sum_{u=1}^{N_\text{UE}}\sum_{k=1}^{N_\text{sub}}}\log\left|\mathbf{I}_{N_\text{r}}+\mathbf{R}_{\text{in},k,u}^{-1}\mathbf{H}_{k,u}\mathbf{F}_{k,u}\left(\mathbf{F}_{k,u}\right)^H\mathbf{H}_{k,u}^H\right|}{N_\text{sub}},
\end{equation}
where 
\begin{equation}
    \mathbf{R}_{\text{in},k,u} = \sum_{i=1, i\neq u}^{N_\text{UE}}\mathbf{H}_{k,u}\mathbf{F}_{k,i}\left(\mathbf{F}_{k,i}\right)^H\mathbf{H}_{k,u}^H+\sigma_\text{n}^2\mathbf{I}_{N_\text{r}}
\end{equation}
is the interference-plus-noise covariance matrix of user $u$.

\subsection{Sensing Signal Model}
\label{subsec: 2-3}
When there are $N_\text{tar}$ targets, by applying again the standard OFDM processing, the received radar signal on subcarrier $k$ and OFDM symbol $l$ is expressed as: 
\begin{equation}
\label{eq: radarRxSignal}
\begin{aligned}
    \mathbf{y}^{\text{Ra}}_{k,l} =& \sum^{N_\text{tar}}_{t=1}\beta_t\mathbf{a}_\text{r}\left(f_k,\theta_t\right)\mathbf{a}_\text{t}^H\left(f_k,\theta_t\right)\mathbf{F}_k\\
    &\cdot\mathbf{s}_{k,l}e^{-j2\pi k\Delta f\tau_t}e^{j2\pi f_{\text{D},t}l(T_\text{sym}+T_\text{CP})} + \mathbf{n}^{\text{Ra}}_{k,l},
\end{aligned}
\end{equation}
where $\mathbf{s}_{k,l}=[\mathbf{s}_{k,l}^{1},\ldots, \mathbf{s}_{k,l}^{N_\text{UE}}]$; $\beta_t$, $\mathbf{a}_\text{r}$, $\mathbf{a}_\text{t}$, $\tau_t$, $f_{\text{D},t}$ are the radar cross-section (RCS), receive steering vector, transmit steering vector, delay, and Doppler shift of target $t$, respectively; and $T_\text{sym}$, $T_\text{CP}$, and $\mathbf{n}^{\text{Ra}}_{k,l}$ are OFDM symbol period,  time duration of cyclic prefix, and sensing noise vector on subcarrier $k$ of OFDM symbol $l$ with variance being $\sigma_\text{n}^{\text{sen}^2}$, respectively. Here, we consider uniform linear array for simplicity. Thus, $[\mathbf{a}_\text{r} \left(f_k,\theta_t\right)]_{n_\text{r}} =e^{-j2\pi f_k(n_\text{r}-1)d_\text{rx}\sin\theta_t}$ and $[\mathbf{a}_\text{t}\left(f_k,\theta_t\right)]_{n_\text{t}} = e^{-j2\pi f_k(n_\text{t}-1)d_\text{tx}\sin\theta_t}$, where
$f_k$ is the subcarrier frequency of subcarrier $k$. Note that extending to other antenna array structure is straightforward. 

We consider the division-FFT-based radar signal sensing procedure proposed in \cite{xiao2024novel} to conduct the target detection along with their range, Doppler velocity, and angle estimations. The idea is to first remove the undesirable precoded symbols $\mathbf{x}_{k,l}=\mathbf{F}_k\mathbf{s}_{k,l}$ from the received signal, and then detect the targets. To do this, we first apply the received beamforming to each subcarrier and angle with the steering vectors, transforming the received signals from the antenna domain to the angle domain. We suppose that the receive beamforming with angle $\theta_m$ is used, leading to: 
\begin{equation}
\begin{aligned}
    &\tilde{y}_{m,k,l} =\mathbf{a}_\text{r}^H\left(f_k,\theta_m\right)\mathbf{y}^{\text{Ra}}_{k,l}
    \approx N_\text{r}^\text{sen}\mathbf{a}^H_{\text{t}}(f_k,\theta_m)\mathbf{x}_{k,l}\\
    &\qquad\cdot\sum_{t\in\mathcal{T}_m} \beta_t e^{j2\pi k\Delta f\tau_t}e^{j2\pi f_{D_t}lT}+\mathbf{a}_\text{r}^H\left(f_k,\theta_m\right)\mathbf{n}_{k,l},
\end{aligned}
\end{equation}
where $\mathcal{T}_m = \{t:|\theta_t - \theta_m| <\frac{\epsilon_\text{a}}{2}, 
         \forall t = 1\ldots, N_\text{tar}\}$ and $\epsilon_\text{a}$ is a tolerable angle approximation error.   Then, we remove the symbols $\mathbf{x}_{k,l}$ from $\tilde{y}_{m,k,l}$ by division:
\begin{equation}
    \begin{aligned}
        &z_{m,k,l}=\tilde{y}_{m,k,l} / \alpha_m\mathbf{a}^H_{\text{t}}(f_k,\theta_m)\mathbf{x}_{k,l} \\
        &=\frac{N_\text{r}^\text{sen}}{\alpha_m}\sum_{t\in\mathcal{T}_m} \beta_t e^{j2\pi k\Delta f\tau_t}e^{j2\pi f_{D_t}lT} +\frac{\mathbf{a}_\text{r}^H\left(f_k,\theta_m\right)\mathbf{n}_{k,l}}{\alpha_m\mathbf{a}^H_{\text{t}}(f_k,\theta_m)\mathbf{x}_{k,l}},
    \end{aligned}
\end{equation}
where $\alpha_{m}=\sqrt{\frac{\sum_{k,l}|\tilde{y}_{m,k,l} / \mathbf{a}^H_{\text{t}}(f_k,\theta_m)\mathbf{x}_{k,l}|^2}{\sum_{k,l}|\tilde{y}_{m,k,l}|^2}}$ is the normalization factor for angle $\theta_m$, making the total received power from each angle remains unchanged after division \cite{xiao2024novel}. Next, we transform the signal from the time-frequency domain to the delay-Doppler domain by applying IFFT and FFT. This creats the delay-Doppler (RD) map given as:
\begin{equation}
\label{eq: RDSignal}
    \begin{aligned}
        Z_{m,k',l'}&=\frac{1}{\sqrt{N_\text{sub}N_\text{sym}}}\sum_{k=1}^{N_\text{sub}}\sum_{l=1}^{N_\text{sym}}z_{m,k,l}e^{j2\pi\frac{k'}{N_\text{sub}}k}e^{-j2\pi\frac{l'}{N_\text{sym}}l}\\
        &\approx\frac{N_\text{r}^\text{sen}\sqrt{N_\text{sub}N_\text{sym}}}{\alpha_m}\sum_{t=\mathcal{T}_{m,k',l'}} \beta_t+N_{m,k',l'},
    \end{aligned}
\end{equation}
where $\mathcal{T}_{m,k',l'} = \{t:|\theta_t-\theta_m|\leq\frac{\epsilon_\text{a}}{2}\land|\frac{\tau_t}{\epsilon_\text{d}} - k'|\leq\frac{\epsilon_\text{d}}{2} \land
         |\frac{f_{\text{D},t}}{\epsilon_\text{D}} - l'|\leq\frac{\epsilon_\text{D}}{2}, 
         \forall t = 1\ldots, N_\text{tar}\}$, $\epsilon_\text{d}=\frac{1}{N_\text{sub}\Delta f}$ is the delay resolution,  $\epsilon_\text{D}=\frac{1}{N_\text{sym}(T_\text{sym}+T_\text{CP})}$ is the Doppler resolution, and 
\begin{equation}
\label{eq: RDNoise}
\begin{aligned}
    N_{m,k',l'}&= \frac{1}{\sqrt{N_\text{sub}N_\text{sym}}}\sum_{k=1}^{N_\text{sub}}\sum_{l=1}^{N_\text{sym}}\frac{\mathbf{a}_\text{r}^H\left(f_k,\theta_m\right)\mathbf{n}_{k,l}}{\alpha_m\mathbf{a}^H_{\text{t}}(f_k,\theta_m)\mathbf{x}_{k,l}} \\
    &\quad\cdot e^{j2\pi\frac{k'}{N_\text{sub}}k}e^{-j2\pi\frac{l'}{N_\text{sym}}l},
\end{aligned}
\end{equation}
is the equivalent RD map Gaussian noise on the delay-Doppler map with zero mean and variance being 
\begin{equation}
\label{eq: RDNoiseDistribution}
    \mathrm{Var}\left(N_{m, k',l'}\right)= \frac{N_\text{r}^\text{sen}\sigma_\text{n}^{\text{sen}^2}}{\alpha_m^2N_\text{sub}N_\text{sym}}\sum_{k=1}^{N_\text{sub}}\sum_{l=1}^{N_\text{sym}}|\mathbf{a}^H_{\text{t}}(f_k,\theta_m)\mathbf{x}_{k,l}|^{-2}.
\end{equation}
Finally, to detect the targets, we apply the cell averaging constant false alarm rate (CA-CFAR) detection separately on the RD map created by using different beamforming angle $\theta_m$, where the RD bin $(k',l')$ with large enough power would be detected as that a target exists. We denote $\mathcal{M}$ as the set of beamforming angles. Then, by applying the above procedure to all beamforming angle $\theta_m,\forall m\in\mathcal{M}$, we can detect targets on different angles with different ranges and Doppler velocities. Furthermore, as we might assume that the resolution on the RD domain is small enough so that there exits at most a single target for a specific range and Doppler velocity pair, whenever the same RD bin is detected on different angles, we only consider the one gives the largest power as the correct angle, matching the point-target assumption at the beginning of this section. Note that such manipulation might not be needed if the point-target assumption is not used, according to \cite{xiao2024novel}. 
%

\subsection{Power Consumption and Energy Efficiency Models}
\label{subsec: 2-4}
When turning on/off the RF chains, the circuit power consumption could be greatly changed. According to \cite{li2020dynamic}, the total power consumption of the considered MIMO-OFDM transmitter can be computed as:
\begin{equation}
    \label{eq: totalPower}
    P_\text{total}\left(\mathbf{F}\right) = \sum_{k=1}^{N_\text{sub}}\|\mathbf{F}_k\|_\text{F}^2 / \eta_\text{PA}+P_\text{BB} + P_\text{RF}\sum_{i=1}^{N_\text{RF}}u\left(\|\mathbf{F}\left(i,:\right)\|_2\right),
\end{equation}
where $N_\text{RF}$ is the number of RF chains and $\eta_\text{PA}$ is the power amplifier efficiency; $P_\text{BB}$ and $P_\text{RF}$ are the power consumption of the baseband precoder and an activated RF chain, respectively; and the function $u(t)$ is the unit step function whose value is $1$ if $t>0$, indicating that a RF chain might be turned off if there is no transmission on that RF chain. Note that the first term in (\ref{eq: totalPower}) is the transmission power consumption, and the third term is the total power consumed by activating RF chains. Subsequently, with (\ref{eq: SE}) and (\ref{eq: totalPower}), the energy efficiency (EE) of the MIMO-OFDM ISAC system is expressed as $EE\left(\mathbf{F}\right) = \frac{R\left(\mathbf{F}\right)}{P_\text{total}\left(\mathbf{F}\right)}$.

\subsection{Transmit Antenna Array Architectures}

In this paper, we study the EE designs under three different transmit antenna array architectures shown in Fig. \ref{fig:AntArt}, which includes the FD precoding, FC hybrid precoding, and PC hybrid precoding architectures. With different antenna array architectures, the aforementioned system and signal models mostly apply and the changes only happen on the mathematical structure of $\mathbf{F}_k$. In the FD precoding, there is no special structure. On the other hand, when hybrid precoding is adopted, $\mathbf{F}_k$ needs to be further split into the baseband digital precoding and analog precoding parts, where the analog pecoding is composed of phase shifters whose exact structures are different for FC and PC hybrid precoding. As shown in Fig. \ref{fig:AntArt}, with the FC architecture, a RF chain is connected to all antennas through independent phase shifters. Differently, with the PC architecture, each RF chain is connected to an independent set of $\frac{N_\text{t}}{N_\text{RF}}$ phase shifters and antennas, constituting a subarray for this RF chain, where $\frac{N_\text{t}}{N_\text{RF}}$ is set to be an integer in this paper for simplicity.

\begin{figure}[t]
\centering
\begin{subfigure}{0.45\textwidth}
    \includegraphics[width=\linewidth]{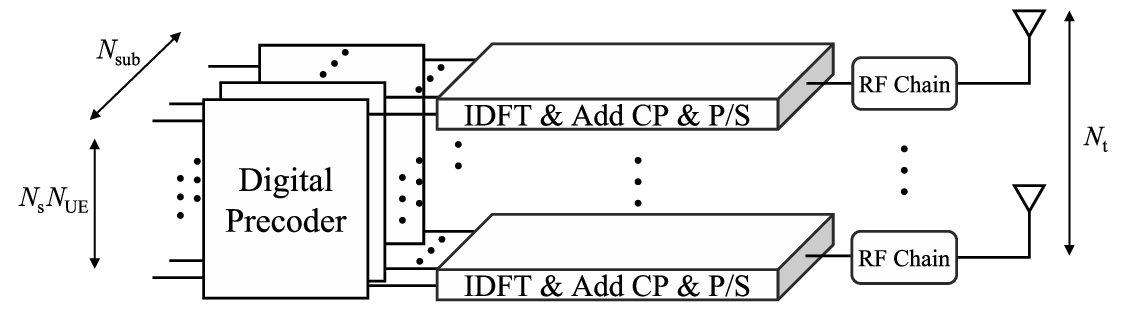}
    \caption{FD precoding architecture.}
    \label{fig:FD}
\end{subfigure}
\hfill
\begin{subfigure}{0.45\textwidth}
    \includegraphics[width=\linewidth]{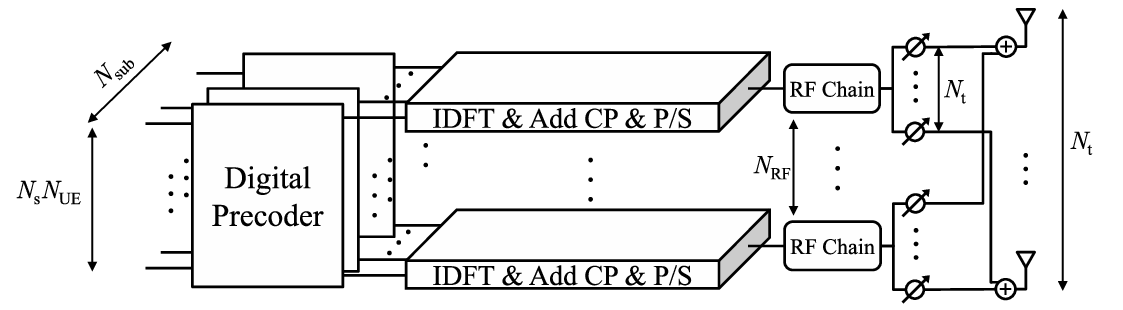}
    \caption{FC hybrid precoding architecture.}
    \label{fig:FC}
\end{subfigure}
\vspace{0.3em}
\begin{subfigure}{0.45\textwidth}
    \includegraphics[width=\linewidth]{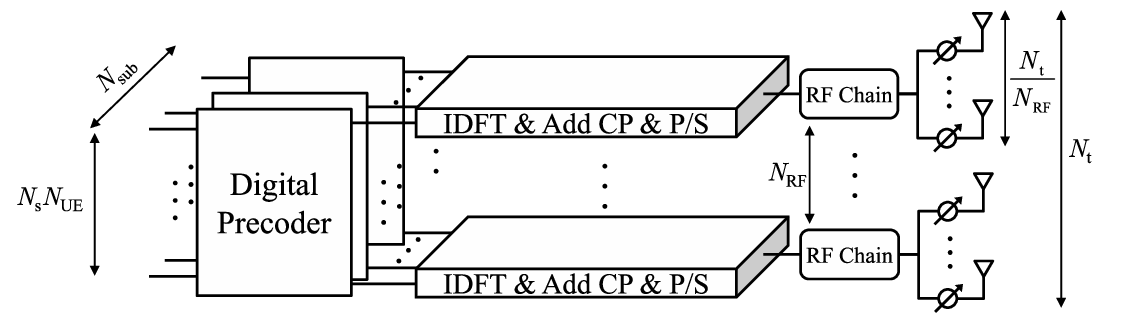}
    \caption{PC hybrid precoding architecture.}
    \label{fig:PC}
\end{subfigure}
\caption{Illustrations of different antenna architectures.}
\vspace{-10pt}
\label{fig:AntArt}
\end{figure}

\section{Energy Efficient Digital Precoding and RF Chain Selection Approaches}
\label{sec_3}

Our goal in this paper is to develop energy efficient precoding and RF chain selection approach for both conventional FD and hybrid precoding architectures. To this end, we in this section start with the FD precoding. We first formulate an EE maximization problem that facilitates the joint optimization of precoding and RF chain selection for the ISAC system. Then, approaches to solve the problem are proposed and discussed.

\subsection{Problem Formulation}

To derive the EE maximization problem for the ISAC system, we need design criteria for both EE and radar sensing. Then, we see that the detection rate under the detection procedure in Sec. \ref{sec: 2} heavily relies on the signal and noise powers of the RD maps. In addition, from  (\ref{eq: RDNoiseDistribution}), it can be seen that the noise power is inversely proportional to $|\mathbf{a}^H_{\text{t}}(f_k,\theta_m)\mathbf{F}_k\mathbf{s}_{k,l}|^2$. Hence, we define the beam power toward angle $\theta$ at subcarrier $k$ as:
\begin{equation}
\label{eq: beampower}
    B_k\left(\mathbf{F}_k, \theta\right) =\|\mathbf{F}_{k}^H\mathbf{a}\left(f_k,\theta\right)\|^2.
\end{equation}
Then, by following the detection probability analysis similar to that in \cite{hsu2021analysis}, the following lemma can be obtained:
\begin{lemma}
\label{thm: powerCon}
Suppose that a target exists at angle $\theta$. Let $P_\text{th}$ be some lower bound satisfying $B_k\left(\mathbf{F}_k, \theta\right)\geq P_\text{th}, \forall k$. When the transmitted data symbols satisfy the i.i.d. standard Gaussian distribution assumption, the detection probability can be improved if the lower bound $P_\text{th}$ is improved. 
\end{lemma}
\begin{proof}
    The analysis follows the RD map analysis in \cite{hsu2021analysis}.
\end{proof}

With Lemma \ref{thm: powerCon}, we consider maximizing EE while satisfying the required detection probability lower bound. Thus, the design problem is formulated as: 
\begin{subequations}
\allowdisplaybreaks
    \label{eq:EE_prob}
\begin{alignat}{2}
    & \max_{\mathbf{F}} &&\quad  EE\left(\mathbf{F}\right) \nonumber\\
    &\;\;\text{s.t.} && \quad B_k\left(\mathbf{F}_k, \theta\right)\geq P_\text{th}, \forall k=1\ldots N_\text{sub}, \forall\theta\in\Theta_\text{tar}, \label{eq: radarCon}\\
    &&& \quad \sum_{k=1}^{N_\text{sub}}\|\mathbf{F}_k\|_F^2 \leq P_\text{tx}, \label{eq: powerCon}
    \end{alignat}
\end{subequations}
where $\Theta_\text{tar}$ is the set containing all the angles of interest, (\ref{eq: radarCon}) helps guaranteeing the sensing performance as indicated by Lemma \ref{thm: powerCon}, and (\ref{eq: powerCon}) is the transmit power constraint. We observe that \eqref{eq:EE_prob} is hard to solve by standard solvers since it includes multiple non-convex terms. Therefore, we analyze the problem and propose an effective solution approach for it below in Sec. \ref{sec: 2}-B.

\subsection{Proposed Iterative Design Approach}

\label{sec_FD_prop}

To derive the design approach, we first need to deal with the non-differentiable unit step function, corresponding to the RF chain selection feature in (\ref{eq: totalPower}), i.e., which RF chains are selected to be turned on or off. Specifically, we observe that $\|\mathbf{F}\left(i, :\right)\|_2$ is non-negative. Thus, we propose to approximate the unit-step function by using the hyperbolic tangent function, i.e., we let $u\left(x\right) \approx \tanh\left(x\right), \forall x\geq0$. Consequently, the approximated total power consumption can be expressed as:
\begin{equation}
    \label{eq: relaxedPtotal}
    \hat{P}_\text{total}\left(\mathbf{F}, \lambda\right) = \sum_{k=1}^{N_\text{sub}}\|\mathbf{F}_k\|_F^2 / \eta_\text{PA}+P_\text{BB} + P_\text{RF}\hat{N}_\text{RF}\left(\mathbf{F},\lambda\right),
\end{equation}
where $\hat{N}_\text{RF}\left(\mathbf{F},\lambda\right) = \sum_{i=1}^{N_\text{RF}}\tanh\left(\lambda\|\mathbf{F}\left(i,:\right)\|_2\right)$, 
and $\lambda$ is the coefficient that controls the shape of the hyperbolic tangent function, in which if $\lambda\rightarrow\infty$, $\hat{P}_\text{total}\rightarrow P_\text{total}$. With (\ref{eq: relaxedPtotal}), the approximate design problem parameterized by $\lambda$ can obtained:
\begin{equation}
\label{pr: hyperbolicRelaxed}
    \begin{aligned}[b]
        & \max_{\mathbf{F}} &&  R\left(\mathbf{F}\right) / \hat{P}_\text{total}\left(\mathbf{F}, \lambda\right)\\
    &\;\;\text{s.t.} && B_k\left(\mathbf{F}_k, \theta\right)\geq P_\text{th}, \forall k,\theta,  \quad\sum_{k=1}^{N_\text{sub}}\|\mathbf{F}_k\|_F^2 \leq P_\text{tx}.
    \end{aligned}
\end{equation} 
Subsequently, since the problem in (\ref{pr: hyperbolicRelaxed}) is a fractional problem, it can be addressed by quadratic transform (QT) proposed in \cite{shen2018fractional} which equivalently reformulate the problem as:
\begin{equation}
\label{pr: FD_QT}
    \begin{aligned}[b]
        & \max_{\mathbf{F}, \mu} &&  2\mu R^\frac{1}{2}\left(\mathbf{F}\right) - \mu^2\hat{P}_\text{total}\left(\mathbf{F}, \lambda\right) \\
    & \;\;\text{s.t.} &&  B_k\left(\mathbf{F}_k, \theta\right)\geq P_\text{th}, \forall k,\theta,\quad \sum_{k=1}^{N_\text{sub}}\|\mathbf{F}_k\|_F^2 \leq P_\text{tx}, 
    \end{aligned}
\end{equation}
where $\mu$ is an auxiliary variable whose optimal solution when given a $\mathbf{F}$ is:
\begin{equation}
\label{eq: optMu}
    \mu^\star = R^\frac{1}{2}\left(\mathbf{F}\right) / \hat{P}_\text{total}\left(\mathbf{F}, \lambda\right),
\end{equation}
where $\mu^\star$ is obtained by solving the first-order optimality condition for the objective function in \eqref{pr: FD_QT}. Note that by substituting $\mu^\star$ into the objective function of (\ref{pr: FD_QT}), (\ref{pr: hyperbolicRelaxed}) indeed can be restored from (\ref{pr: FD_QT}). Then, since $R\left(\mathbf{F}\right)$ is still nonconvex, we address it by using the WMMSE method in \cite[Theorem 17]{SCA} which converts the weighted sum-rate maximization problem to a WMMSE problem. We then provide Lemma \ref{le: WMMSE} as:
\begin{lemma}
\label{le: WMMSE}
When given a set of precoding matrices $\mathbf{F}_{k,u},\forall u$ of subcarrier $k$, the SE $R_{k,u}$ of subcarrier $k$ of user $u$ is:
\begin{equation}
    \begin{aligned}
        R_{k,u}\left(\mathbf{F}_k\right) &= \log\left|\mathbf{I}_{N_\text{r}}+\mathbf{R}_{\rm{in},k,u}^{-1}\mathbf{H}_{k,u}\mathbf{F}_{k,u}\left(\mathbf{F}_{k,u}\right)^H\mathbf{H}_{k,u}^H\right| \\
        &=\max_{\mathbf{U}_{k,u}, \mathbf{W}_{k,u}}\log\left|\mathbf{W}_{k,u}\right| - \mathrm{tr}\left(\mathbf{W}_{k,u}\mathbf{E}_{k,u}\right)+N_\text{r},
    \end{aligned}
\end{equation}
where $\mathbf{W}_{k,u}\in\mathbb{C}^{N_\text{r}\times N_\text{r}}$ is an auxiliary matrix and
\begin{equation}
    \begin{aligned}
        \mathbf{E}_{k,u} = & \left(\mathbf{I}-\mathbf{U}^H_{k,u}\mathbf{H}_{k,u}\mathbf{F}_{k,u}\right)\left(\mathbf{I}-\mathbf{U}^H_{k,u}\mathbf{H}_{k,u}\mathbf{F}_{k,u}\right)^H + \\
        &\sum_{\substack{i=1\\i\neq u}}^{N_\text{UE}}\mathbf{U}_{k,u}^H\mathbf{H}_{k,u}\mathbf{F}_{k,i}\mathbf{F}_{k,i}^H\mathbf{H}_{k,u}^H\mathbf{U}_{k,u} + \sigma_\text{n}^2\mathbf{U}_{k,u}^H\mathbf{U}_{k,u}.
    \end{aligned}
\end{equation}
is the mean-square-error matrix.
\end{lemma}
\begin{proof}
    See \cite[Theorem 17]{SCA}.
\end{proof}

By Lemma \ref{le: WMMSE} and the fact that the optimal $\mu$ should be positive, the problem in (\ref{pr: FD_QT}) can be reformulated as:
\begin{equation}
\label{pr: WMMSE}
    \begin{aligned}[b]
        & \max_{\mathbf{F}, \mu, \mathbf{U}, \mathbf{W}} &&  2\mu \tilde{R}^\frac{1}{2}\left(\mathbf{F}, \mathbf{U}, \mathbf{W}\right) - \mu^2 \hat{P}_\text{total}\left(\mathbf{F}, \lambda\right) \\
    & \;\;\;\;\text{s.t.} &&  B_k\left(\mathbf{F}_k, \theta\right)\geq P_\text{th}, \forall k,\theta, \quad \sum_{k=1}^{N_\text{sub}}\|\mathbf{F}_k\|_F^2 \leq P_\text{tx}, 
    \end{aligned}
\end{equation}
where $\tilde{R}\left(\mathbf{F}, \mathbf{U}, \mathbf{W}\right) =  \frac{1}{N_\text{sub}}\sum_{u=1}^{N_\text{UE}}\sum_{k=1}^{N_\text{sub}}\log\left|\mathbf{W}_{k,u}\right| - \mathrm{tr}\left(\mathbf{W}_{k,u}\mathbf{E}_{k,u}\right)+N_\text{r}$;
$\mathbf{U}_{k,u}$ and $\mathbf{W}_{k,u}$ are auxiliary matrices.
Subsequently, to resolve the non-convexity in the hyperbolic tangent functions of $\hat{P}_\text{total}$ and $B_k\left(\mathbf{F}_k, \theta\right)\geq P_\text{th}$ in (\ref{pr: WMMSE}), we exploit the successive convex approximation (SCA) \cite{SCA} to obtain a convexified problem, and then solve the convexified problem iteratively. Specifically, by using the first-order approximations obtained from a feasible point $\mathbf{F}^r$, we can have $\tilde{P}_\text{total}\left(\mathbf{F}, \lambda;\mathbf{F}^r\right)= \sum_{k=1}^{N_\text{sub}}\|\mathbf{F}_k\|_\text{F}^2 / \eta_\text{PA}+P_\text{BB} + 
P_\text{RF}\tilde{N}_\text{RF}\left(\mathbf{F},\lambda;\mathbf{F}^r\right)$, 
where 
\begin{equation}
\label{eq: linearizedNRF}
\begin{aligned}    &\tilde{N}_\text{RF}\left(\mathbf{F},\lambda;\mathbf{F}^r\right) = \sum_{i=1}^{N_\text{RF}}\Biggl(\tanh\left(\lambda \|\mathbf{F}^r\left(i,:\right)\|_2\right)+\\
    &  \lambda\mathrm{sech}^2\left(\lambda \|\mathbf{F}^r\left(i,:\right)\|_2\right)\left(\|\mathbf{F}\left(i,:\right)\|_2 - \|\mathbf{F}^r\left(i,:\right)\|_2\right)\Biggr),
\end{aligned}
\end{equation}
and 
\begin{equation}
\label{eq: linearizedRadCon}
\begin{aligned}
    &\tilde{B}_k\left(\mathbf{F}_k, \theta;\mathbf{F}_k^r\right) = \|\left(\mathbf{F}_{k}^{r}\right)^H\mathbf{a}\left(f_k,\theta\right)\|^2 \\
    &+ 4\mathrm{Re}\biggl\{\mathrm{tr}\left(\left(\mathbf{F}_{k}^{r}\right)^H\mathbf{a}\left(f_k,\theta\right)\mathbf{a}^H\left(f_k,\theta\right)\left(\mathbf{F}_k - \mathbf{F}_k^r\right)\right)\biggl\},
\end{aligned}
\end{equation}
where $\mathrm{Re}(x)$ is the operation that takes the real part of $x$. The convexified problem is then given as:
\begin{equation}
\label{pr: linearized}
    \begin{aligned}[b]
        & \max_{\mathbf{F}, \mu, \mathbf{U}, \mathbf{W}} &&  2\mu\tilde{R}^\frac{1}{2}\left(\mathbf{F}, \mathbf{U}, \mathbf{W}\right) - \mu^2\tilde{P}_\text{total}\left(\mathbf{F}, \lambda;\mathbf{F}^r\right) \\
    & \;\;\;\;\text{s.t.} && \tilde{B}_k\left(\mathbf{F}_k, \theta;\mathbf{F}_k^r\right)\geq P_\text{th}, \forall k,\quad \sum_{k=1}^{N_\text{sub}}\|\mathbf{F}_k\|_F^2 \leq P_\text{tx}.
    \end{aligned}
\end{equation}
Finally, we note that the optimal auxiliary variables can be obtained by solving the optimality conditions of $ \mathbf{U}_{k,u}$ and $\mathbf{W}_{k,u}$, leading to:
\begin{equation}
    \begin{aligned}
            \mathbf{U}_{k,u}^\star &= \left(\sum_{i=1}^{N_\text{UE}}\mathbf{H}_{k,u}\mathbf{F}_{k,i}\mathbf{F}_{k,i}^H\mathbf{H}_{k,u}^H+\sigma_\text{n}^2\mathbf{I}\right)^{-1}\mathbf{H}_{k,u}\mathbf{F}_{k,u},\\
    \mathbf{W}_{k,u}^\star &= \mathbf{E}_{k,u}^{-1}
    \label{eq: optW}.
    \end{aligned}
\end{equation}
Therefore, with a given feasible solution $\mathbf{F}^r$ obtained at iteration $r$, we can set $\mu$, $\mathbf{U}$, and $\mathbf{W}$ to their optimal values, denoted as $\mu^r$, $\mathbf{U}^r$, and $\mathbf{W}^r$, respectively, by using $\mathbf{F}^r$ and expressions in \eqref{eq: optMu} and \eqref{eq: optW}. It follows that the convexified problem can be cast as:
\begin{equation}
\label{pr: subproblemFD}
    \begin{aligned}
        & \max_{\mathbf{F}} &&  2\mu^r \tilde{R}^\frac{1}{2}\left(\mathbf{F}, \mathbf{U}^r, \mathbf{W}^r\right) - (\mu^{r})^2 \tilde{P}_\text{total}\left(\mathbf{F}, \lambda;\mathbf{F}^r\right) \\
    & \;\;\text{s.t.} && \tilde{B}_k\left(\mathbf{F}_k, \theta;\mathbf{F}_k^r\right)\geq P_\text{th}, \forall k,\theta, \\
    &&& \sum_{k=1}^{N_\text{sub}}\|\mathbf{F}_k\|_F^2 \leq P_\text{tx}.
    \end{aligned}
\end{equation}
which is a convex problem that can be solved by standard solvers. Then, since \eqref{pr: subproblemFD} is solvable by using standard convex solvers, we can iteratively solve (\ref{pr: subproblemFD}) until convergence to obtain an effective solution for the approximated joint precoding and RF chain selection problem in (\ref{pr: hyperbolicRelaxed}).

Finally, recall that $\lambda$ is the coefficient that controls the shape of the hyperbolic tangent function and larger $\lambda$ gives a better approximation for the power consumption model. However, a large $\lambda$ indeed renders the slope of the hyperbolic tangent function be too steep, resulting in certain numerical issue. Therefore, we propose to sequentially update $\lambda$ in each iteration so that the value of $\lambda$ also increases from a small value to a large value that tightens the approximation gap in the end. Indeed, when $\lambda$ is small, our approach can more flexibly move  the precoding along the hyperbolic tangent function so that the objective function of the approximate problem can be better improved. On the other hand, as $\lambda$ increases, the state of each RF chain gradually becomes locked, making the iterative approach eventually converges to a 0-1-like solution. As a consequence, when given a threshold $\epsilon$, the final on-off decisions for RF chains can be easily obtained by using rounding. The overall approach is summarized in Alg. \ref{alg_FDInnerLoop}.

\begin{algorithm}[t]
\caption{Proposed Optimization Approach for $\mathbf{F}$}
\label{alg_FDInnerLoop}
\begin{algorithmic}[1]
    \State {Given a feasible $\mathbf{F}_0^{\text{outer}}$, $\nu>1$, $\lambda=1$, and $s=0$}
    \While {$s\leq R_{\text{outer}}$ or update is within $\epsilon$}
    \State {$\mathbf{F}_0^{\text{inner}}=\mathbf{F}_s^{\text{outer}}$. $r=0$.}
    \While {$r\leq R_{\text{inner}}$}
        \State {Update $\mu$ by using (\ref{eq: optMu}).}
        \State {Update $\mathbf{U}$ and $\mathbf{W}$ by using (\ref{eq: optW}).}
        \State {Solve (\ref{pr: subproblemFD}) to obtain $\mathbf{F}_{r+1}^{\text{inner}}$. $r=r+1$. }
    \EndWhile
    \State{$\mathbf{F}_{s+1}^{\text{outer}}=\mathbf{F}_{R_{\text{inner}}}^{\text{inner}}$. $s=s+1$ and $\lambda=\nu\lambda$.}
    \EndWhile
    \State{$\mathbf{F}^*$ is obtained by using rounding with $\mathbf{F}_{R_{\text{outer}}}$.}
\end{algorithmic}
\end{algorithm}
 
\subsection{Proposed Alternative Design Approaches}

\label{sec: FD_refScheme}

The proposed approach in Sec. \ref{sec_FD_prop} relies on relaxing the integer variables associated with the RF chain selection via using a hyperbolic tangent function, thereby enhancing the tractability of the problem. However, alternative methods can be devised that avoid such relaxation, albeit at the expense of increased computational complexity. In line with this perspective, this section explores both brute-force-based and greedy-based design approaches as alternative design approaches, in which both approaches are still based on the QT and WMMSE transformations discussed in Sec. \ref{sec_FD_prop}, while the main difference lies on how to deal with the RF chain selection. To this end,  \eqref{eq:EE_prob} can be reformulated as:
\begin{equation}
\label{pr: benchmarkReform}
    \begin{aligned}
    & \max_{\mathbf{F}, \mathbf{A}} &&  \frac{\displaystyle{\sum_{u=1}^{N_\text{UE}}\sum_{k=1}^{N_\text{sub}}}\log\left|\mathbf{I}_{N_\text{r}}+\mathbf{R}_{\text{in},k,u}^{-1}\mathbf{H}_{k,u}\mathbf{A}\mathbf{F}_{k,u}\mathbf{F}_{k,u}^H\mathbf{A}\mathbf{H}_{k,u}^H\right|}{N_\text{sub}\left(\sum_{k=1}^{N_\text{sym}}\|\mathbf{A}\mathbf{F}_k\|_F^2 / \eta_\text{PA}+P_\text{BB} + P_\text{RF}\mathrm{tr}\left(\mathbf{A}\right)\right)}\\
    &\;\;\text{s.t.} && \|\mathbf{F}_{k}^H\mathbf{A}\mathbf{a}\left(f_k,\theta\right)\|_2^2\geq P_\text{th}, \forall k=1\ldots N_\text{sub}, \forall\theta\in\Theta_\text{tar},\\
    &&& \sum_{k=1}^{N_\text{sub}}\|\mathbf{A}\mathbf{F}_k\|_F^2 \leq P_\text{tx},\\
    &&& \mathbf{A}\left(i,i\right) \in \{0,1\}, \forall i, \mathbf{A}\left(i,j\right)=0, \forall i,j, i\neq j,\\
\end{aligned}
\end{equation}
where $\mathbf{R}_{\text{in},k,u} = \sum_{i=1, i\neq u}^{N_\text{UE}}\mathbf{H}_{k,u}\mathbf{A}\mathbf{F}_{k,i}\mathbf{F}_{k,i}^H\mathbf{A}\mathbf{H}_{k,u}^H+\sigma_\text{n}^2\mathbf{I}_{N_\text{r}}$ and $\mathbf{A}$ is a diagonal binary selection matrix whose diagonal elements are indicators of the RF chain selection.

To develop the brute-force-based approach, the idea is to first develop the precoding design when given a selection matrix $\mathbf{A}$, and then examines all possible RF selections to find the one providing the best EE. To this end, we see that when $\mathbf{A}$ is given, the circuit power consumption is then fixed. As a consequence, by following the same prcedure in \eqref{pr: hyperbolicRelaxed}-\eqref{pr: subproblemFD}, we can obtain the convexified precoder design problem having the similar form of \eqref{pr: subproblemFD}, with the difference lies on that the effective communication channels $\mathbf{H}_{k,u},\forall k,u$ and radar steering vectors $\mathbf{a}(f_k,\theta),\forall k,\theta$ need to be replaced by the effective communication channels $\mathbf{A}\mathbf{H}_{k,u},\forall k,u$ and effective radar steering vectors $\mathbf{A}\mathbf{a}(f_k,\theta),\forall k,\theta$, respectively. In addition, the approximate power consumption expression $\tilde{P}_\text{total}$ should be replaced by exact power consumption, given as $\sum_{k=1}^{N_\text{sym}}\|\mathbf{A}\mathbf{F}_k\|_F^2 / \eta_\text{PA}+P_\text{BB} + P_\text{RF}\mathrm{tr}\left(\mathbf{A}\right)$. With SCA method, when given a selection matrix $\mathbf{A}$, the convexified problem is iteratively solved until convergence for finding the precoders corresponding to the selection $\mathbf{A}$. Finally, the brute-force-based approach examines all possible RF selections along with their corresponding precoders obtained by aforementioned procedure to attain the one providing the best EE. This approach then with very high complexity as the number of possible selections to check is $2^{N_{\text{t}}}-1$.

Since the brute-force-based approach has extremely high complexity, we discuss the greedy-based approach that has lower complexity as follows. At the beginning, the greedy-based approach considers that all RF chains are active, and then use the aforementioned precoding design approrach to obtain the corresponding precoding and EE, recorded as $EE^\star$. Subsequently, the approach starts to greedily turn off the RF chain. Specifically, in each step, the approach generates a list containing all the active RF chains in a randomized order. Following the order, the approach turns off the RF chain and obtain the corresponding precoding and EE value $EE^{\text{cur}}$. Then, if $EE^{\text{cur}} > EE^\star$, the approach turns this RF chain off, and then proceed to the next step, where a random list is generated again; otherwise, the approach does not turn the RF chain off, and then proceed to test the next RF chain in the list. Such procedure continues until all the active RF chains in the list are tested and none of them are turned off or there is only a single activated RF chain left. 

\subsection{Extension to Spectral Efficiency and Power Consumption Tradeoff Design}

Our proposed design approaches can be easily extended to obtain SE and power consumption tradeoff designs. To this end, we first formulate the SE and power consumption tradeoff design problem as:
\begin{equation}
\label{pr:SEPowTrade}
    \begin{aligned}[b]
        & \max_{\mathbf{F}, \mu} &&  \omega_1 R\left(\mathbf{F}\right) - \omega_2P_\text{total}\left(\mathbf{F}, \lambda\right) \\
    & \;\;\text{s.t.} &&  B_k\left(\mathbf{F}_k, \theta\right)\geq P_\text{th}, \forall k,\theta,\quad \sum_{k=1}^{N_\text{sub}}\|\mathbf{F}_k\|_F^2 \leq P_\text{tx}, 
    \end{aligned}
\end{equation} 
where $\omega_1$ and $\omega_2$ are tradeoff factors. We then observe that \eqref{pr:SEPowTrade} has a expression similar to that in \eqref{pr: FD_QT}, where the major difference lies in that $R\left(\mathbf{F}\right)$ does not need to take the square-root. However, $R\left(\mathbf{F}\right)$ without taking the square-root indeed is more tractable. Thus, we can directly follow the procedure in Sec. \ref{sec_FD_prop}, where we first approximate the unit-step function using the hyperbolic tangent function, and then apply Lemma \ref{le: WMMSE} together with SCA in \eqref{eq: linearizedNRF} and \eqref{eq: linearizedRadCon} to derive an iterative algorithm that has a similar structure to Alg. \ref{alg_FDInnerLoop}. This leads to the tradeoff design approach, where different tradeoff points can be obtained by having different values of $\omega_1$ and $\omega_2$.

\section{Energy Efficient Design under Fully-Connected Architecture}

\label{sec: FC}

In this section, we propose the hybrid precoding and RF chain selection design approach under the FC architecture. 

\subsection{System Model and Problem Formulation}

When considering the FC structure, the signal models in \ref{sec: 2} should should be modified by replacing $\mathbf{F}_{k,u}$ with $\mathbf{F}_\text{RF}\mathbf{F}_{\text{BB},k}^u$, where $\mathbf{F}_\text{RF}$ is the analog precoder whose dimensionality is $N_\text{t}\times N_\text{RF}$, with $|\mathbf{F}_\text{RF}\left(i,j\right)| = 1$, and $\mathbf{F}_{\text{BB},k}^u$ is the baseband digital precoder for subcarrier $k$ of user $u$. Then, with the diagonal binary selection matrix $\mathbf{A}$, SE of the system is expressed as:
\begin{equation}
\begin{aligned}
    \label{eq: SEFC}
    &R_{\text{FC}}\left(\mathbf{F}_\text{RF}, \mathbf{F}_{\text{BB}}, \mathbf{A}\right) = \frac{1}{N_\text{sub}}\sum_{u=1}^{N_\text{UE}}\sum_{k=1}^{N_\text{sub}}\\
    &\log\left|\mathbf{I}_{N_\text{r}}+\mathbf{R}_{\text{in},k,u}^{-1}\mathbf{H}_{k,u}\mathbf{F}_\text{RF}\mathbf{A}\mathbf{F}_{\text{BB},k}^{u}\left(\mathbf{F}_{\text{BB},k}^{u}\right)^H\mathbf{A}\mathbf{F}_\text{RF}^H\mathbf{H}_{k,u}^H\right|,
\end{aligned}
\end{equation}
where
\begin{equation}
    \mathbf{R}_{\text{in},k,u} = \sum_{i\neq u}^{N_\text{UE}}\mathbf{H}_{k,u}\mathbf{F}_\text{RF}\mathbf{A}\mathbf{F}_{\text{BB},k}^{i}\left(\mathbf{F}_{\text{BB},k}^{i}\right)^H\mathbf{A}\mathbf{F}_\text{RF}^H\mathbf{H}_{k,u}^H+\sigma_\text{n}^2\mathbf{I}_{N_\text{r}}
\end{equation}
is the interference-plus-noise covariance matrix for subcarrier $k$ of user $u$. When considering the hybrid precoding, the power consumption of phase shifters in the analog precoder should be included in the power consumption model. Consequently, since each RF chain in the FC structure is connected to $N_\text{t}$ phase shifters, the power consumption model is:
\begin{equation}
\begin{aligned}
        &P_\text{total}^\text{FC}\left(\mathbf{F}_\text{RF}, \mathbf{F}_{\text{BB}}, \mathbf{A}\right) =\\ &\sum_{k=1}^{N_\text{sub}}\|\mathbf{F}_\text{RF}\mathbf{A}\mathbf{F}_{\text{BB},k}\|_F^2 / \eta_\text{PA}+P_\text{BB} 
        + \left(P_\text{RF} + N_\text{t}P_\text{PS}\right)\mathrm{tr}\left(\mathbf{A}\right),
\end{aligned}
\end{equation}
where $P_\text{PS}$ is the power consumption of a phase shifter and $\mathbf{F}_{\text{BB},k}=[\mathbf{F}_{\text{BB}, k}^1,\ldots,\mathbf{F}_{\text{BB},k}^{N_\text{UE}}]$. With the above results, EE of the system considering the FC architecture is:
\begin{equation}
\label{eq: FCEE}
\begin{aligned}
    EE_\text{FC}\left(\mathbf{F}_\text{RF}, \mathbf{F}_{\text{BB}},\mathbf{A}\right)=\frac{R_{\text{FC}}\left(\mathbf{F}_\text{RF}, \mathbf{F}_{\text{BB}}, \mathbf{A}\right)}{P_\text{total}^\text{FC}\left(\mathbf{F}_\text{RF}, \mathbf{F}_{\text{BB}}, \mathbf{A}\right)}.
\end{aligned}
\end{equation}

To formulate the EE maximization problem for hybrid precoding with FC architecture, we derivethe beam power toward angle $\theta$ at subcarrier $k$ is derived as: 
\begin{equation}
\label{eq: BPFC}
    B_k\left(\mathbf{F}_\text{RF}, \mathbf{F}_{\text{BB}, k}, \mathbf{A}, \theta\right) =\|\mathbf{F}_\text{RF}^H\mathbf{A} \mathbf{F}_{\text{BB}, k}^H\mathbf{a}\left(f_k,\theta\right)\|_2^2.
\end{equation}
Therefore, the energy efficient hybrid precoding and RF chain selection design problem under FC architecture is:
\begin{equation}
\label{pr: FC}
    \begin{aligned}
    & \max_{\mathbf{F}_\text{RF}, \mathbf{F}_{\text{BB}}, \mathbf{A}} &&  EE_\text{FC}\left(\mathbf{F}_\text{RF}, \mathbf{F}_{\text{BB}},\mathbf{A}\right)\\
    &\quad\text{ s.t.} && B_k\left(\mathbf{F}_\text{RF}, \mathbf{F}_{\text{BB}, k}, \mathbf{A}, \theta\right)\geq P_\text{th}, \forall k,\theta,\\
    &&& \sum_{k=1}^{N_\text{sub}}\|\mathbf{F}_\text{RF}\mathbf{A}\mathbf{F}_{\text{BB},k}\|_F^2 \leq P_\text{tx},\\
    &&& |\mathbf{F}_\text{RF}\left(i,j\right)| = 1, \forall i,j,\\
    &&& \mathbf{A}\left(i,i\right) \in \{0,1\}, \forall i, \mathbf{A}\left(i,j\right)=0, \forall i\neq j.
\end{aligned}
\end{equation}

\subsection{Proposed Design Approach}

When considering the FC architecture, due to symmetry, it can be observed that which RF chains are activated does not change the solution of the (\ref{pr: FC}), as long as the the number of active RF chains is identical. Therefore, the number of candidates for RF chain selection can be reduced from $2^{N_\text{RF}-1}$ to only $N_\text{RF}$, where each candidate represents a distinct number of active RF chains. We thus define the candidate set of RF chain selection as $\mathcal{A} = \big\{\mathbf{A}_1,\mathbf{A}_2,\ldots,\mathbf{A}_{N_\text{RF}}\big\}$, 
where $\mathbf{A}_n$ is a diagonal selection matrix whose first $n$ diagonal elements are equal to $1$, while all remaining diagonal elements are $0$. With this, our approach is to search through all possible candidates in $\mathcal{A}$, and then selects the best one as the solution. The remaining is to obtain the energy efficient hybrid precoding for each candidate in $\mathcal{A}$.

To this end, when considering the candidate $\mathbf{A}_n$, (\ref{pr: FC}) can be reformulated as:
\begin{equation}
\label{pr: FC_fixA}
    \begin{aligned}
        & \max_{\mathbf{F}_\text{RF}, \mathbf{F}_\text{BB}} &&  EE_\text{FC}\left(\mathbf{F}_\text{RF}, \mathbf{F}_{\text{BB}},\mathbf{A}_n\right)\\
        &\;\;\text{s.t.} && \|\mathbf{F}_{\text{BB},k}^H\mathbf{A}_n\mathbf{F}_\text{RF}^H\mathbf{a}\left(f_k,\theta\right)\|_2^2\geq P_\text{th}, \forall k,\theta,\\
        &&& \sum_{k=1}^{N_\text{sub}}\|\mathbf{F}_\text{RF}\mathbf{A}_n\mathbf{F}_{\text{BB},k}\|_F^2 \leq P_\text{tx}, |\mathbf{F}_\text{RF}\left(i,j\right)| = 1, \forall i,j.
    \end{aligned}
\end{equation}
To solve \eqref{pr: FC_fixA}, we first notice that when considering the hybrid precoding design with FC architecture, it has been shown in \cite{yu2016alternating} and many other situations that if the corresponding FD precoding can be obtained, an effective hybrid precoding can be attained by solve a matching problem that aims to match the hybrid precoding to the FD precoding. Following this principle, we propose to consider solving: 
\begin{equation}
\label{pr: FCmatching}
    \begin{aligned}
        &\min_{\mathbf{F}'_\text{RF}, \mathbf{F}'_\text{BB}} && \|\mathbf{F}_n^\text{opt} - \mathbf{F}'_\text{RF}\mathbf{F}'_\text{BB}\|_F \\
        & \quad\text{s.t.} && \|\mathbf{F}'_\text{RF}\mathbf{F}'_{\text{BB}}\|_F^2 \leq P_\text{tx}, |\mathbf{F}'_\text{RF}\left(i,j\right)| = 1, \forall i,j,
    \end{aligned}
\end{equation}
where $\mathbf{F}'_\text{RF} = \mathbf{F}_\text{RF}\left(:,1 : n\right)$ is the matrix containing the first $n$ columns of the original $\mathbf{F}_\text{RF}$; $\mathbf{F}'_\text{BB} = \mathbf{F}_\text{BB}\left(1: n, :\right)$ is the matrix containing the first $n$ columns of the original $\mathbf{F}_\text{BB}$; and $\mathbf{F}_n^\text{opt}$ is a FD precoder obtained by solving \eqref{pr: WMMSE} using one of our proposed approaches in Sec. \ref{sec_3}, under which the activated RF chains are already determined by $\mathbf{A}_n$, i.e., the first $n$ RF chains are activated. Note that for those RF chains and corresponding phase shifters are not activated, their values are simply set to $0$. In addition, while directly solving \eqref{pr: WMMSE} can yield an effective $\mathbf{F}_n^\text{opt}$, we shall modify the power consumption model in \eqref{pr: WMMSE} as $\sum_{k=1}^{N_\text{sub}}\|\mathbf{F}_k\|_F^2 / \eta_\text{PA}+P_\text{BB} + n\left(P_\text{RF} + N_\text{t}P_\text{PS}\right)$ to better match that of the FC architecture. Solving the modified problem to obtain $\mathbf{F}_n^\text{opt}$ indeed can lead to a better solution for \eqref{pr: FC_fixA}, as the adjusted power model more accurately reflects the system's actual power consumption.

To solve (\ref{pr: FCmatching}), we follow the similar approach provided in \cite{yu2016alternating}. Specifically, we first ignore the power constraint that couples $\mathbf{F}_\text{RF}'$ and $\mathbf{F}_\text{BB}'$ for a moment, and thus $\mathbf{F}_\text{RF}'$ and $\mathbf{F}_\text{BB}'$ becomes separable in constraints. As a consequence, we can use block coordinate descent (BCD) method \cite{BCD} to iteratively update $\mathbf{F}_\text{RF}'$ and $\mathbf{F}_\text{BB}'$ until convergence. To this end, we see that when fixing $\mathbf{F}'_\text{RF}$ to solve $\mathbf{F}'_\text{BB}$, the optimal $\mathbf{F}'_\text{BB}$ is the least-square solution given as $\mathbf{F}_\text{BB}^{'} = \left(\mathbf{F}_\text{RF}^{'H}\mathbf{F}'_\text{RF}\right)^{-1} \mathbf{F}_\text{RF}^{'H}\mathbf{F}_\text{opt}$.
On the other hand, when fixing $\mathbf{F}'_\text{BB}$ and elements in $\mathbf{F}'_\text{RF}$ except for column $q$, we can obtain the subproblem as:
\begin{equation}
    \label{pr: FCmatching_RFsubproblem}
    \begin{aligned}
        \min_{\mathbf{f}_{\text{RF},q}} \left\|\mathbf{F}_n^\text{opt} - \sum_{p=1}^{n}\mathbf{f}_{\text{RF},p}  \mathbf{f}^T_{\text{BB},p}\right\|_F \quad \text{s.t.} \quad |\mathbf{f}_{\text{RF},q}\left(i\right)| = 1, \forall i,
    \end{aligned}
\end{equation}
where $\mathbf{f}_{\text{RF},p} = \mathbf{F}_\text{RF}'\left(:, p\right)$ is column $p$ of $\mathbf{F}_\text{RF}'$ and $\mathbf{f}_{\text{BB},p}^T = \mathbf{F}_\text{BB}'\left(p, :\right)$ is column $p$ of $\mathbf{F}_\text{BB}'$. It follows that (\ref{pr: FCmatching_RFsubproblem}) can be further reformulated as: 
\begin{equation}
    \begin{aligned}
        \max_{\mathbf{f}_{\text{RF},j}} \mathrm{Re}\Bigg\{\mathbf{f}_{\text{fix},q}\mathbf{f}^H_{\text{RF},q}\Bigg\}\quad \text{s.t.} \quad|\mathbf{f}_{\text{RF},q}\left(i\right)| = 1, \forall i,
    \end{aligned}
\end{equation}
where $\mathbf{f}_{\text{fix},q} = \left(\mathbf{F}_\text{opt}-\sum_{p=1,p\neq q}^{n}\mathbf{f}_{\text{RF},p}  \mathbf{f}^T_{\text{BB},p}\right)\mathbf{f}^*_{\text{BB},q}$. Hence, the  optimal $\mathbf{f}_{\text{RF},q}$ can be obtained as:
\begin{equation}
\label{eq: FCoptRF}
    \mathbf{f}_{\text{RF},q} \left(i\right)= e^{j\angle\mathbf{f}_{\text{fix},q}\left(i\right)},\forall i,
\end{equation}
where $\angle\mathbf{f}_{\text{fix},q}\left(i\right)$ is the phase of $\mathbf{f}_{\text{fix},q}\left(i\right)$. With \eqref{eq: FCoptRF}, the column $q$ of the analog precoding matrix can be updated. Subsequently, by iteratively updating different columns of the analog precoding matrix using \eqref{eq: FCoptRF}, we can update the whole analog precoding when given the digital precoding $\mathbf{F}'_\text{BB}$. Finally, by iteratively updating $\mathbf{F}_\text{BB}'$ and $\mathbf{F}_\text{RF}'$ until convergence, we can obtain the effective design of the hybrid precoder. Since our precoder design problem is to maximize energy efficiency, the power constraint is usually satisfied. Since the above procedure ignore the power constraint, we need to conduct the power normalization if $\|\mathbf{F}_\text{RF}'\mathbf{F}_{\text{BB}}'\|_F^2 > P_\text{tx}$. We note that it has been shown in \cite[Lemma 1]{yu2016alternating} that such normalization does not increase the matching error by more than a factor of two.


Finally, by solving \eqref{pr: FC_fixA} with different $\mathbf{A}_n$ using the aforementioned approach, we can obtain solutions under different RF chain selections. We then compare the EE of them to obtain the most energy efficient hybrid precoding and RF chain selection design for the FC architecture. We note that similar to the case of fully-digital precoding design, our approach here can be directly extended to obtain the SE-power consumption tradeoff design approach for hybrid precoding under fully-connected architecture.  


\section{Energy Efficient Design under Partially-Connected Architecture}

\label{sec: PC}

In this section, we propose the hybrid precoding and RF chain selection design approach under the PC architecture.

\subsection{System Model and Problem Formulation}

When considering the hybrid precoding and RF chain selection with PC architecture, the system model in Sec. \ref{sec: FC} can be mostly reused, and the only difference is that $\mathbf{F}_\text{RF}$ under the PC architecture should be a block diagonal matrix, in which only elements $(j-1)\frac{\mathbf{N}_\text{t}}{\mathbf{N}_\text{RF}}+1,...,(j-1)\frac{\mathbf{N}_\text{t}}{\mathbf{N}_\text{RF}}+\frac{\mathbf{N}_\text{t}}{\mathbf{N}_\text{RF}}-1$ in column $j$ of $\mathbf{F}_\text{RF}$ are non-zero. As a result, we let $\bar{\mathbf{F}} = \mathbf{F}_\text{RF}\mathbf{F}_\text{BB}$. It follows that when RF chain $i$ is turned off, elements in row $\left(i-1\right)\frac{N_\text{t}}{N_\text{RF}}+1$ to row $i\frac{N_\text{t}}{N_\text{RF}}$ would be $0$, leading to that $u\left(\bigg\|\bar{\mathbf{F}}\left(\left(i-1\right)\frac{N_\text{t}}{N_\text{RF}}+1:i\frac{N_\text{t}}{N_\text{RF}},:\right)\bigg\|_F\right)=0$; otherwise, $u\left(\bigg\|\bar{\mathbf{F}}\left(\left(i-1\right)\frac{N_\text{t}}{N_\text{RF}}+1:i\frac{N_\text{t}}{N_\text{RF}},:\right)\bigg\|_F\right)=1$.
Hence, the number of active RF chains can be computed as:
\begin{equation}
\begin{aligned}
        &\bar{N}_\text{RF}^\text{PC} = \sum_{i=1}^{N_\text{RF}}u\left(\bigg\|\bar{\mathbf{F}}\left(\left(i-1\right)\frac{N_\text{t}}{N_\text{RF}}+1:i\frac{N_\text{t}}{N_\text{RF}},:\right)\bigg\|_F\right),
\end{aligned}
\end{equation}
leading to the total power consumption given as: 
\begin{equation}
    \begin{aligned}
    P_\text{total}^\text{PC}\left(\mathbf{F}_\text{RF}, \mathbf{F}_{\text{BB}}\right) &= \frac{N_\text{t}}{N_\text{RF}} \sum_{k=1}^{N_\text{sub}}\|\mathbf{F}_{\text{BB},k}\|_F^2 / \eta_\text{PA} \\
    &+P_\text{BB}+ \left(P_\text{RF} + \frac{N_\text{t}}{N_\text{RF}}P_\text{PS}\right) \bar{N}_\text{RF}.
    \end{aligned}
\end{equation}
Moreover, the SE for the hybrid precoding under PC architecture $R_\text{PC}\left(\mathbf{F}_\text{RF}, \mathbf{F}_{\text{BB}}\right)$ follows the similar expression of (\ref{eq: SEFC}), with the modification that the binary selection matrix $\mathbf{A}$ is removed. 
Similarly, the expression for the the beam power toward angle $\theta$ at subcarrier $k$ modified from \eqref{eq: BPFC} is given as:
\begin{equation}
    B_k\left(\mathbf{F}_\text{RF}, \mathbf{F}_{\text{BB}, k}, \theta\right) =\|\mathbf{F}_\text{RF}^H\mathbf{F}_{\text{BB}, k}^H\mathbf{a}\left(f_k,\theta\right)\|_2^2.
\end{equation}
With the above, the hybrid precoding design and RF chain selection problem with PC architecture is formulated as:
\begin{equation}
\label{pr: PCproblem}
    \begin{aligned}
    & \max_{\mathbf{F}_\text{RF}, \mathbf{F}_{\text{BB}}} &&  EE_\text{
    PC}\left(\mathbf{F}_\text{RF}, \mathbf{F}_{\text{BB}}\right)=\frac{R_{\text{PC}}\left(\mathbf{F}_\text{RF}, \mathbf{F}_{\text{BB}}\right)}{P_\text{total}^\text{PC}\left(\mathbf{F}_\text{RF}, \mathbf{F}_{\text{BB}}\right)} \\
    &\;\;\;\,\text{s.t.} && B_k\left(\mathbf{F}_\text{RF}, \mathbf{F}_{\text{BB},k}, \theta\right)\geq P_\text{th}, \forall k,\theta, \\
    &&& \frac{N_\text{t}}{N_\text{RF}} \sum_{k=1}^{N_\text{sub}}\|\mathbf{F}_{\text{BB},k}\|_F^2 \leq P_\text{tx}, \mathbf{F}_\text{RF}\in\mathcal{F},
    \end{aligned}
\end{equation}
where $\mathcal{F}$ is the set containing all the possible analog precoders satisfying the PC architecure.

\subsection{Proposed Design Approach}

To obtain the design under PC architecture, we again propose to first obtain a FD precoding and RF chain selection by solving the following problem:
\begin{equation}
\label{pr: PCeqvalentFD}
    \begin{aligned}
        & \max_{\mathbf{F}} &&  \frac{R\left(\mathbf{F}\right)}{\sum_{k=1}^{N_\text{sym}}\|\mathbf{F}_k\|_F^2 / \eta_\text{PA}+P_\text{BB} + \left(P_\text{RF} + \frac{N_\text{t}}{N_\text{RF}}P_\text{PS}\right)\bar{N}_\text{RF}\left(\mathbf{F}\right)}\\
        &\;\;\text{s.t.} && \|\mathbf{F}_k^H\mathbf{a}\left(f_k,\theta\right)\|_2^2\geq P_\text{th}, \forall k,\theta, \sum_{k=1}^{N_\text{sub}}\|\mathbf{F}_k\|_F^2 \leq P_\text{tx},\\
    \end{aligned}
\end{equation}
where 
\begin{equation}
    \bar{N}_\text{RF}\left(\mathbf{F}\right) = \sum_{i=1}^{N_\text{RF}}u\left(\bigg\|\mathbf{F}\left(\left(i-1\right)\frac{N_\text{t}}{N_\text{RF}}+1:i\frac{N_\text{t}}{N_\text{RF}},:\right)\bigg\|_F\right).
\end{equation}
It can be observed that \eqref{pr: PCeqvalentFD} is simply the modified problem of \eqref{eq:EE_prob}, with the modification to let the power consumption model match the PC architecture. Then, since \eqref{pr: PCeqvalentFD} has a similar formulation as that in \eqref{eq:EE_prob}, it can be solved by using the similar procedure presented in Sec. \ref{sec_FD_prop}, where the main difference lies in that the unit-step function in $\bar{N}_\text{RF}\left(\mathbf{F}\right)$ is now approximated to match the power consumption model of the PC architecture, given as:
\begin{equation}
    \bar{N}_\text{RF}\left(\mathbf{F}\right) \approx \hat{N}_\text{RF}\left(\mathbf{F}\right) = \sum_{i=1}^{N_\text{RF}}\tanh\left(\lambda\big\|\mathbf{F}^{(i)}\big\|_F\right),
\end{equation}
where $\mathbf{F}^{(i)} = \mathbf{F}\left(\left(i-1\right)\frac{N_\text{t}}{N_\text{RF}}+1:i\frac{N_\text{t}}{N_\text{RF}},:\right)$.

With the the FD precoding matrix $\mathbf{F}_\text{PC}^\text{opt}$ obtained by using the modified Alg. \ref{alg_FDInnerLoop} as described above, we then design the analog and digital precoders. We note that by having $\mathbf{F}_\text{PC}^\text{opt}$, which RF chains are turned off have already been determined. Thus, the remaining is to derive the corresponding analog and digital precoders by solving the following problem:
\begin{equation}
\label{pr: PCmatching}
    \begin{aligned}
        &\min_{\mathbf{F}_\text{RF}, \mathbf{F}_\text{BB}} && \|\mathbf{F}_\text{PC}^\text{opt} - \mathbf{F}_\text{RF}\mathbf{F}_\text{BB}\|_F \\
        & \quad\text{s.t.} && \frac{N_\text{t}}{N_\text{RF}} \sum_{k=1}^{N_\text{sub}}\|\mathbf{F}_{\text{BB},k}\|_F^2 \leq P_\text{tx}, \mathbf{F}_\text{RF}\in\mathcal{F}',
    \end{aligned}
\end{equation}
where $\mathcal{F}'$ contains the feasible analog precoding matrices that satisfy the RF chain selection suggested by $\mathbf{F}_\text{PC}^\text{opt}$. As a result, if $\mathbf{F}_\text{PC}^\text{opt}$ suggests to turn off RF chain $i$, $\mathcal{F}'$ would only contains analog precoding matrices that do not use RF chain $i$.

To solve \eqref{pr: PCmatching}, we again temporarily ignore the power constraint to decouple $\mathbf{F}_\text{RF}$ and $\mathbf{F}_\text{BB}$ in the constraint. Then, due that each RF chain in the PC architecture is connected to an independent set of phase shfiters, when ignoring the power constraint, we can decompose the problem into the individual subproblems, where each subproblem independently design the parts of $\mathbf{F}_\text{RF}$ and $\mathbf{F}_\text{BB}$ corresponding to the specific RF chain. As a result, when considering RF chain $i$, we can obtain the following subproblem $i$, expressed as:
\begin{equation}
\label{pr: PCmatchingSub}
    \begin{aligned}
        &\min_{\mathbf{f}_{\text{RF},i}, \mathbf{F}_\text{BB}\left(i, :\right)} && \left\|\mathbf{F}_\text{PC}^{\text{opt},i} - \mathbf{f}_{\text{RF},i}\mathbf{F}_\text{BB}\left(i, :\right)\right\|_F \\
        & \quad\;\:\,\text{s.t.} && |\mathbf{f}_{\text{RF},i}\left(j\right)| = 1,  j =1,\ldots, \frac{N_\text{t}}{N_\text{RF}},
    \end{aligned}
\end{equation}
where $\mathbf{F}_\text{PC}^{\text{opt},i} = \mathbf{F}_\text{PC}^{\text{opt}}\left(i\frac{N_\text{t}}{N_\text{RF}}+1 : (i+1)\frac{N_\text{t}}{N_\text{RF}},:\right)$, and $\mathbf{f}_{\text{RF},i}\in\mathbb{C}^{\frac{N_\text{t}}{N_\text{RF}}}$ is the phase shifters of subarray $i$ in the analog precoder which is the $i$th block diagonal matrix of $\mathbf{F}_\text{RF}$, namely, $\mathbf{F}_\text{RF}=\text{blkdiag}(\mathbf{f}_{\text{RF},1},\mathbf{f}_{\text{RF},2},...,\mathbf{f}_{\text{RF},N_\text{RF}})$.

The subproblem in \eqref{pr: PCmatchingSub} can then be solved by using the BCD method that iteratively updates $\mathbf{f}_{\text{RF},i}$ and $\mathbf{F}_\text{BB}\left(i, :\right)$. Specifically, when fixing $\mathbf{f}_{\text{RF},i}$, the optimal digital precoder corresponding to RF chain $i$ is given as:
\begin{equation}
\label{eq: PCoptFBB}
    \mathbf{F}_\text{BB}\left(i,:\right) = \left(\mathbf{f}_{\text{RF},i}^{H}\mathbf{f}_{\text{RF},i}\right)^{-1} \mathbf{f}_{\text{RF},i}^{H}\mathbf{F}_\text{PC}^{\text{opt},i} = \frac{N_\text{t}}{N_\text{RF}}\mathbf{f}_{\text{RF},i}^{H}\mathbf{F}_\text{PC}^{\text{opt},i}.
\end{equation}
On the other hand, when fixing $ \mathbf{F}_\text{BB}\left(i,:\right)$, subproblem \eqref{pr: PCmatchingSub} can be equivalently written as: 
\begin{equation}
    \begin{aligned}
    \label{eq:opt_subarray_analog}
        &\max_{\mathbf{f}_{\text{RF},j}} && \mathrm{Re}\left\{\mathrm{tr}\left(\mathbf{F}_\text{PC}^{\text{opt},i}\mathbf{F}_\text{BB}^H\left(i, :\right)\mathbf{f}^H_{\text{RF},i}\right)\right\}\\
        & \;\;\text{s.t.} && |\mathbf{f}_{\text{RF},i}\left(j\right)| = 1, \forall j,
    \end{aligned}
\end{equation}
where the objective function is the real part of the inner product of $\mathbf{F}_\text{PC}^{\text{opt},i}\mathbf{F}_\text{BB}^H$ and $\mathbf{f}_{\text{RF},i}$. Hence, the optimal solution for \eqref{eq:opt_subarray_analog} is $\mathbf{f}_{\text{RF},i} = e^{j\angle\left(\mathbf{F}_\text{opt}^{(i)}\mathbf{F}_\text{BB}^H\right)}$.
By iteratively updating $\mathbf{F}_\text{BB}\left(i,:\right)$ and $\mathbf{f}_{\text{RF},i}$ for all $i$ whose RF chains are suggested to be turned on in $\mathbf{F}_\text{PC}^{\text{opt}}$ until convergence, we can obtain the effective design of the hybrid precoding and RF chain selection for the system with PC architecture. We stress that as mentioned above, the hybrid precoding corresponding to different RF chains can be separately designed with our proposed approach. In addition, if the power constraint in \eqref{pr: PCmatching} is violated, we again normalize the power of $\mathbf{F}_\text{BB}$ to obtain $\|\mathbf{F}_\text{RF}\mathbf{F}_{\text{BB}}\|_F^2 = P_\text{tx}$. Finally, we again note that similar to the previous cases, our approach here can be directly extended to obtain the corresponding SE-power consumption tradeoff design approach.

\section{Convergence and Complexity Analysis}

In this section, the convergence behavior and complexity of our proposed approaches are analyzed.

\subsection{Convergence Analysis}

\label{sec: FD_complexity}

Since the proposed hybrid precoding designs are largely based on the proposed FD design approaches in Sec. \ref{sec_3} and the solution approaches for matching problems in \eqref{pr: FCmatching} and \eqref{pr: PCmatching} are based on standard BCD procedure that guarantees convergence \cite{BCD}, our convergence analysis in this subsection focuses on proposed FD precoding design approach in Sec, \ref{sec_FD_prop}. Specifically, from Alg. \ref{alg_FDInnerLoop}, we can see that the proposed precoding and RF chain selection approach includes both the outer and inner loops. Furthermore, although $\lambda$ could in practice be dynamically increased during the iterations, for the purpose of analysis, we assume that a suitable deterministic value of $\lambda$ is used throughout the iterations. We will later explain why the overall approach could stop when $\lambda\rightarrow\infty$. 
To show the convergence of Alg. \ref{alg_FDInnerLoop} when fixing $\lambda$, our idea is to show that the value of the objective function value in \eqref{pr: FD_QT}, namely, $\hat{EE}\left(\mathbf{F}^{r+1},\mu,\lambda\right)=2\mu R^\frac{1}{2}\left(\mathbf{F}\right) - \mu^2\hat{P}_\text{total}\left(\mathbf{F}, \lambda\right)$, with a given $\lambda$ would increase in each iteration. In other words, we would like to show that when using Alg. \ref{alg_FDInnerLoop}, the following is satisfied for each iteration: $\hat{EE}\left(\mathbf{F}^{r+1},\mu^{r+1},\lambda\right) \geq \hat{EE}\left(\mathbf{F}^r,\mu^{r},\lambda\right)$.

To this end, we let 
\begin{equation}
    \tilde{EE}'\left(\mathbf{F}, \mu, \mathbf{U}, \mathbf{W}\right) =  2\mu \tilde{R}^\frac{1}{2}\left(\mathbf{F}, \mathbf{U}, \mathbf{W}\right) - \mu^2 \hat{P}_\text{total}\left(\mathbf{F}, \lambda\right)
\end{equation}
and let 
\begin{equation}
    \tilde{EE}\left(\mathbf{F}, \mu, \mathbf{U}, \mathbf{W}; \mathbf{F}^r\right) =  2\mu \tilde{R}^\frac{1}{2}\left(\mathbf{F}, \mathbf{U}, \mathbf{W}\right) - \mu^2 \tilde{P}_\text{total}\left(\mathbf{F}, \lambda;\mathbf{F}^r\right).
\end{equation}
Then, the following inequalities must be satisfied:
\begin{equation}
\label{eq:converge_proof}
    \begin{aligned}
        \hat{EE}\left(\mathbf{F}^{r+1},\mu^{r+1},\lambda\right) & \overset{(a)}{=}   \tilde{EE}'\left(\mathbf{F}^{r+1}, \mu^{r+1}, \mathbf{U}^{r+2}, \mathbf{W}^{r+2}\right)\\
        & \overset{(b)}{\geq}   \tilde{EE}\left(\mathbf{F}^{r+1}, \mu^{r+1}, \mathbf{U}^{r+1}, \mathbf{W}^{r+1}; \mathbf{F}^r\right)\\ 
        &  \overset{(c)}{\geq} \tilde{EE}\left(\mathbf{F}^r, \mu^{r+1}, \mathbf{U}^{r+1}, \mathbf{W}^{r+1}; \mathbf{F}^r\right) \\
        & \overset{(d)}{=} \hat{EE}\left(\mathbf{F}^r,\mu^{r},\lambda\right),  
    \end{aligned}
\end{equation}
where $(a)$ and $(b)$ in \eqref{eq:converge_proof} are satisfied because
\begin{equation}
\begin{aligned}
    \hat{EE}\left(\mathbf{F}^{r+1},\mu^{r+1},\lambda\right) &= \max_{\mathbf{U},\mathbf{W}} \;\; \tilde{EE}\left(\mathbf{F}^{r+1}, \mu^{r+1}, \mathbf{U}, \mathbf{W} \right) \\
    &\geq\tilde{EE}\left(\mathbf{F}^{r+1}, \mu^{r+1}, \mathbf{U}^{r+1}, \mathbf{W}^{r+1}; \mathbf{F}^{r+1}\right)
\end{aligned}
\end{equation}
according to the adopted WMMSE method and that SCA method requires the convexified objective function to be lower bound of the original objective function; $(c)$ in \eqref{eq:converge_proof} is satisfied because \eqref{pr: subproblemFD} is a convex problem whose optimal solution is obtainable; 
and finally, $(d)$ in \eqref{eq:converge_proof} is satisfied because
\begin{equation}
    \begin{aligned}
        &\tilde{EE}\left(\mathbf{F}^r, \mu^{r+1}, \mathbf{U}^{r+1}, \mathbf{W}^{r+1}; \mathbf{F}^r\right)=\hat{EE}\left(\mathbf{F}^{r},\mu^{r+1},\lambda\right)\\
        &\geq \hat{EE}\left(\mathbf{F}^{r},\mu^{r},\lambda\right),
    \end{aligned}
\end{equation}
as $\mu^{r+1} =   R^\frac{1}{2}\left(\mathbf{F}^r\right) / \hat{P}_\text{total}\left(\mathbf{F}^r, \lambda\right)$ is the optimal solution. Finally, since \eqref{eq:converge_proof} shows that $\hat{EE}$ can be monotonically increased in each iteration and there must exist a EE upper bound, Alg. \ref{alg_FDInnerLoop} must converge when fixing a $\lambda$. 

If $\lambda\rightarrow\infty$, the slope of the hyperbolic tangent function becomes increasingly steep, eventually approaching $90$ degrees, making the solution  easily trapped at certain point. Thus, Alg. \ref{alg_FDInnerLoop} should automatically stop when $\lambda$ is large enough.


\subsection{Complexity Analysis}

Here, the complexity of the proposed FD precoding and RF selection approaches is first analyzed. Based on this analysis, the complexity of the hybrid precoding and RF selection approaches is then presented.

\subsubsection{Complexity for the Pproposed FD Precoding and RF Chain Selection Approaches}

To analyze the complexity of the proposed iterative approach in Alg. \ref{alg_FDInnerLoop}, we see that the computational complexity of the inner loop is dominated by solving \eqref{pr: benchmarkReform} using the convex solver, as other steps are with closed-form expressions. Thus, supposing that the interior-point method is used to solve the convex problem, the complexity of each iteration of the inner loop is $\mathcal{O}\left(\left(N_\text{t}N_\text{s}N_\text{UE}N_\text{sub}\right)^{3.5}\log\left(\epsilon^{-1}\right)\right)$, where $N_\text{t}N_\text{s}N_\text{UE}N_\text{sub}$ is the total number variables to optimize and $\epsilon$ is the required accuracy \cite{wright1997primal}. We then denote the required number for the inner-loop iterations to converge as $R_{\text{inner}}$ and the required number of outer-loop iterations to converge as $R_{\text{outer}}$. The complexity of our proposed design approach in Sec. \ref{sec_FD_prop} is derived as $\mathcal{O}\left(R_{\text{outer}}R_{\text{inner}}\left(N_\text{t}N_\text{s}N_\text{UE}N_\text{sub}\right)^{3.5}\log\left(\epsilon^{-1}\right)\right)$.

Subsequently, we analyze the complexity of brute-force-based and greedy-based approaches in Sec. \ref{sec: FD_refScheme}. We see that when given as an RF chain selection, the subsequent precoding design follows the the inner-loop steps of Alg. \ref{alg_FDInnerLoop}. Thus, the above complexity analysis directly applies. Then, since the brute-force-based approach requires $2^{N_\text{t}}-1$ trials on the RF chain selection, its complexity is $\mathcal{O}\left(2^{N_\text{t}}R_{\text{inner}}\left(N_\text{t}N_\text{s}N_\text{UE}N_\text{sub}\right)^{3.5}\log\left(\epsilon^{-1}\right)\right)$. On the other hand, with the greedy-based approaches, the worst-case number of trials is at the order of $\frac{N_\text{t}^2+N_\text{t}}{2}$, which happens when eventually only $1$ RF chain is selected, and in each stage, the RF chains are never turned off until the last order of that stage. Hence, the complexity of the greedy-based approach is $\mathcal{O}\left(\frac{N_\text{t}^2+N_\text{t}}{2}R_{\text{inner}}\left(N_\text{s}N_\text{UE}N_\text{sub}\right)^{3.5}\log\left(\epsilon^{-1}\right)\right)$. By comparing between complexity of different approaches, we see that the main difference lies on the number of required trials. By our empirical experience, $R_{\text{outer}}$ in our proposed approach is basically some constant irrespective to other system parameters. Thus, our proposed approach in general has lower complexity than the brute-force-based and greedy-based approaches.

\subsubsection{Complexity for the Pproposed Hybrid Precoding and RF Chain Selection Approaches}

The complexity of our design approaches for FC and PC architectures are composed of two parts. The first part is to find the optimal reference FD precoding used for constructing the matching problem for the hybrid precoding design, and the second part is to solve the constructed matching problems. Then, since the first part adopts our proposed approach in Sec. \ref{sec_FD_prop}, the complexity simply follows the aforementioned procedure, leading again to $\mathcal{O}\left(R_{\text{outer}}R_{\text{inner}}\left(N_\text{t}N_\text{s}N_\text{UE}N_\text{sub}\right)^{3.5}\log\left(\epsilon^{-1}\right)\right)$ for both FC and PC architectures. On the other hand, the complexity analysis for the second part follows the similar procedure as that in \cite{yu2016alternating}. Nevertheless, the second part of the solution framework indeed adopts some closed-form solutions, while the first part requires solving the convex problem, the overall complexity is dominated by the first part, and thus the analysis results for the second part are omitted for brevity.

\section{Computer Simulations}
\label{sec: numericalResults}

\subsection{Simulation Setups}

In this section, we evaluate the proposed designs by examining the tradeoff between EE and detection probability with fixed false alarm rates. Unless otherwise indicated, we consider $f_c=73$ GHz, $\Delta f=240$ KHz, $N_\text{sym}=16$, $N_\text{cp}=8$, $N_\text{s}=2$, $d_\text{tx}=d_\text{rx}=\frac{\lambda}{2}$, $N_r=2$,  $N_\text{r}^\text{sen}=16$, $N_\text{UE}=2$ and $\eta_\text{PA} = 1$. In addition, the communication SNR $\rho_\text{c}/\sigma_\text{n}^2=10$ dB and $64$-QAM is adopted. The power consumption parameters of hardware devices are with $P_\text{RF}=300$mW, $P_\text{BB} = 200$mW, and $P_\text{PS}=50$mW \cite{li2020dynamic}. We adopt the $73$ GHz channel model proposed in \cite{samimi20163}. Also, $\lambda=1$ is adopted for initializing our proposed approaches and the update rate is $\nu= 3\sim 5$. 
We consider two setups for simulations, and the detailed parameters are provided in Table \ref{tab: parameter}, where $P_\text{FA}$ is the required false alarm rate. In addition, The target of setup 1 is with range $156$ m, Doppler velocity $-61$ m/s, and angle $27$ degree; targets of setup 2 are with ranges $78$ m and $370$ m, Doppler velocities $-99$ m/s and $74$ m/s, and angles $-3$ degree and $27$ degree, respectively. 

For comparisons, we compare between the proposed iterative, brute-force-based, and greedy-based approaches in Sec. \ref{sec: FD_refScheme}, labeled as ``proposed'', ``brute force'', and ``greedy'', respectively. In addition, we also compare with random and all-selection approaches, where the random selection randomly selects RF chains to be turned on to use subject to the sensing requirement can be satisfied, while the all-selection turns on all RF chains. Note that the aforementioned random and all-selection approaches only determine the RF chain selection results, and thus they still need to use our proposed precoding design approach in Sec. \ref{sec_FD_prop}, making the comparisons to focus on validating the efficacy of our RF chain selection. Furthermore, we compare our approaches with some other reference schemes in the literature. For the FD and FC cases, we compare with the approaches in \cite{kaushik2022green}, \cite{nguyen2023multiuser} and \cite{liao2025design}; and for the PC case, we compare with approaches in \cite{kaushik2022green}. Note that since the design in \cite{kaushik2022green} is for narrowband MIMO ISAC systems, to compare with it, we conduct some modifications which adopt equal power allocation among all subcarriers. Finally, we note that due to high computational complexity, the brute-force-based scheme is not used in setup 2. In addition, for the FC architecture, since our proposed method already exhaustively searches all possible antenna configurations, there is no need to compare with the corresponding brute-force-based and proposed greedy-based approaches. When considering using the proposed brute-force-based and greedy-based approaches with the PC architecture, we simply let $\mathbf{F}_\text{PC}^\text{opt}$ in \eqref{pr: PCmatching} to be designed by using the brute-force-based and greedy-based approaches in Sec. \ref{sec: FD_refScheme}, while the remaining design procedure is the same.

\begin{table}
\centering
\caption[Simulation Setups]{Simulation Setups}
\label{tab: parameter}
\begin{tabular}{|c|c|c|c|c|c|c|}
\hline
 & \multicolumn{3}{|c|}{Setup 1} & \multicolumn{3}{|c|}{Setup 2} \\ \hline
Structure & FD & FC & PC & FD & FC & PC \\ \hline
$N_\text{t}$ & 8 & \multicolumn{2}{|c|}{16} & \multicolumn{3}{|c|}{32} \\ \hline
$N_\text{RF}$ & 8 & 4,\,8 & 4,\,8 & 32 & 8,\,16 & 8,\,16 \\ \hline
$N_\text{sub}$ & \multicolumn{3}{|c|}{4} & \multicolumn{3}{|c|}{32} \\ \hline
$\beta_\text{t}^2/\sigma_\text{n}^{\text{sen}^2}$ & $-15$\,dB & \multicolumn{2}{|c|}{$-20$\,dB} & \multicolumn{3}{|c|}{$-25$\,dB} \\ \hline
$N_\text{tar}$ & \multicolumn{3}{|c|}{1} & \multicolumn{3}{|c|}{2} \\ \hline
$P_\text{FA}$ & \multicolumn{3}{|c|}{$10^{-2}$} & \multicolumn{3}{|c|}{$10^{-3}$} \\ \hline
\end{tabular}
\end{table}


\subsection{Simulation Results}

We first consider setup 1 and evaluate the EE-detection rate tradeoff under the FD architecture in Fig. \ref{fig: FDsmall_tradeoff}. Results show that the brute-force-based approach achieves the best tradeoff performance. In addition, we see that both our proposed iterative approach and greedy-based approach perform closely to the brute-force-based approach, validating the effectiveness of our approaches. It should be noted that the brute-force-based approach has notoriously high complexity, while proposed greedy-based and our iterative approach are with much lower complexity. Finally, we see that all our proposed approaches outperform the reference schemes in the literature, showing the superiority of our proposed approaches.

Fig. \ref{fig: FDsmall_NRF} shows the number of selected RF chains under the FD architecture for setup 1. We see that when the sensing requirement $P_\text{th}$ is larger, our approaches tend to enable more RF chains to satisfy the sensing requirement. In addition, we see that the number of activated RF chains for the greedy-based approach is very close to that of the brute-force-based approach, implying that the greedy-based approach can perform more closely to the optimal approach, while our iterative approach behaves slightly different. However, since our precoding approach can adapt to different numbers of activated RF chains, the overall EE-detection performance of our proposed iterative approach is not too different from those of the brute-force-based and greedy-based approaches. Finally, we see also that the random-based approach selects more RF chains with respect to $P_\text{th}$. This is because the random-based approach still needs to satisfy the sensing requirement, forcing it to select more RF chains when $P_\text{th}$ is larger. In Fig. \ref{fig: FDlarge_tradeoff}, we evaluate our approaches of FD architecture in setup 2, and due to complexity issue, the brute-force-based approach is not evaluated in setup 2. Results again validate that our approaches are more superior than the reference schemes. Also, results show that our proposed iterative and greedy-based approaches have almost identical performance.


\begin{figure}[t]
\centering
    \includegraphics[width=0.8\linewidth]{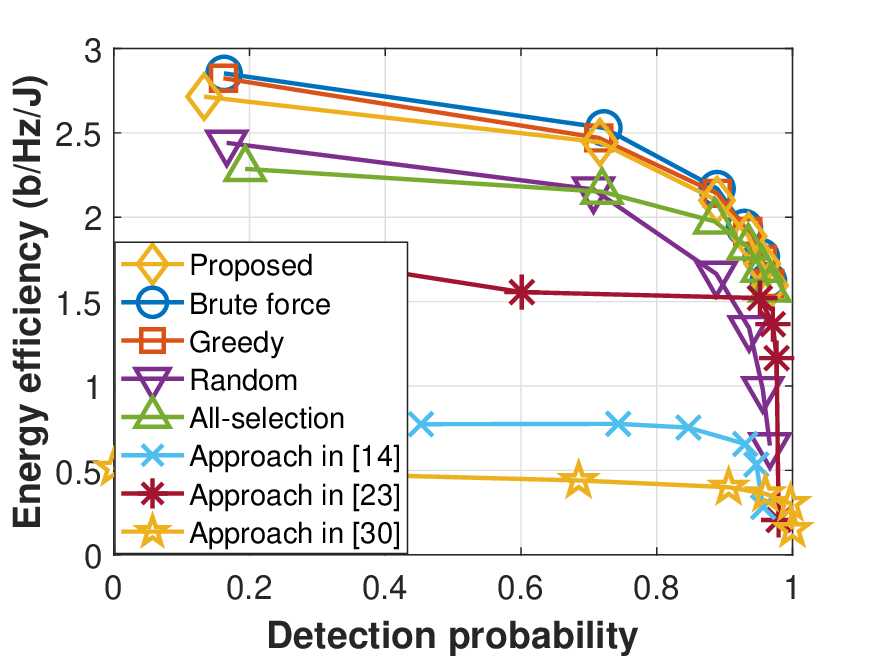} 
    \vspace{-5pt}
    \caption{EE-detection probability tradeoff for approaches with FD architecture in setup 1.}
    \label{fig: FDsmall_tradeoff}
    \vspace{-10pt}
\end{figure}

\begin{figure}[t]
\centering
    \includegraphics[width=0.77\linewidth]{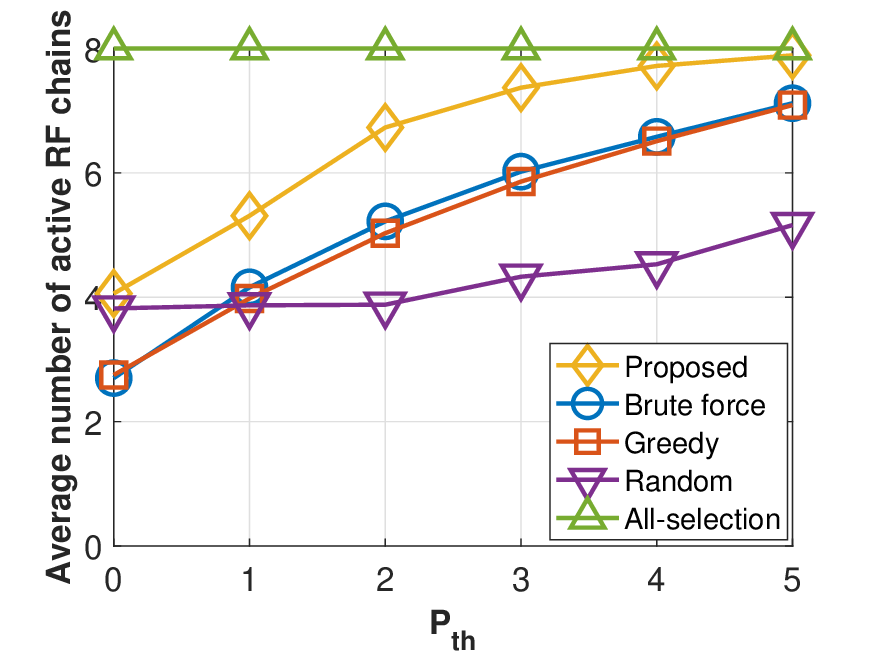}
    \vspace{-5pt}
    \caption{Number of activated RF chains as a function of $P_\text{th}$ for approaches with FD architecture in setup 1.}
    \label{fig: FDsmall_NRF}
    \vspace{-10pt}
\end{figure}

\begin{figure}[t]
\centering
    \includegraphics[width=0.8\linewidth]{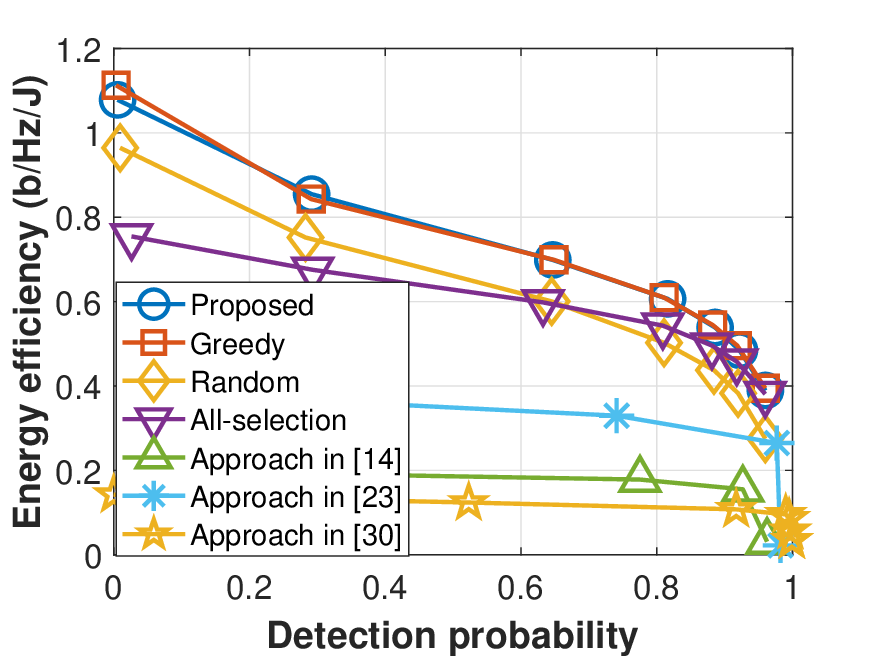} 
    \vspace{-5pt}
    \caption{EE-detection probability tradeoff for approaches with FD architecture in setup 2.}
    \label{fig: FDlarge_tradeoff}
    \vspace{-10pt}
\end{figure}

In Fig.~\ref{fig: FCsmall_tradeoff}, we evaluate our hybrid precoding and RF chain selection approach under FC architecture in setup 1. Note that in the FC architecture, our main proposed approach is identical to the brute-force-based approach. Hence, from the results, we see that our proposed approach consistently outperforms all other reference schemes in both cases where $N_\text{RF}=4$ and $N_\text{RF}=8$. In addition, we also see that the EE-detection probability tradeoff of our approach in both $N_\text{RF}=4$ and $N_\text{RF}=8$ cases are similar, showing that our proposed approach can automatically adjust the number of activated RF chains to the suitable number even if the setups of the FC architecture are different. On the contrary, all other reference schemes cannot provide such good behavior. As a result, the EE performance of the reference schemes are in general worse when $N_\text{RF}=8$ is adopted, as they commonly use more number of RF chains than necessary. In Fig.~\ref{fig: FClarge_tradeoff}, we evaluate our approaches of FC architecture in setup 2. Results again validate that our approach is more superior than other reference approaches and that our approach can appropriately adjust the number of activated RF chains. In Fig. \ref{fig: PCsmall_tradeoff}, we evaluate the performance of our proposed designs under PC architecture for setup 1. Results show that our approaches can outperform the reference schemes. In addition, all other observations under FC architecture can be observed again.

\begin{figure}[!t]
    \begin{subfigure}{0.5\textwidth}
    \centering
        \includegraphics[width=0.8\linewidth]{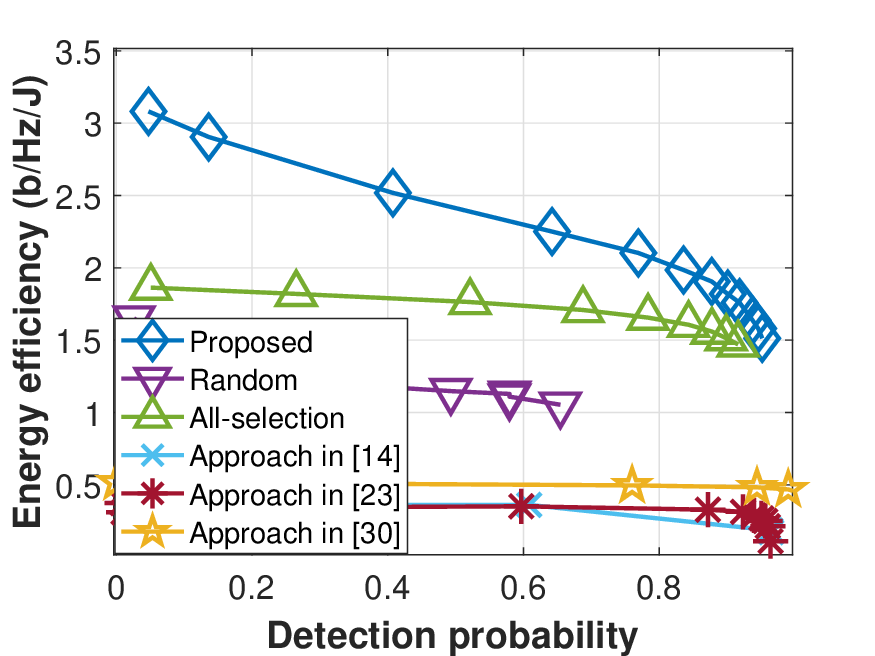}
        \caption{Number of RF chains is $4$.}
        \label{fig: FCsmall_4RF_tradeoff}
    \end{subfigure}
    \begin{subfigure}{0.5\textwidth}
    \centering
        \includegraphics[width=0.8\linewidth]{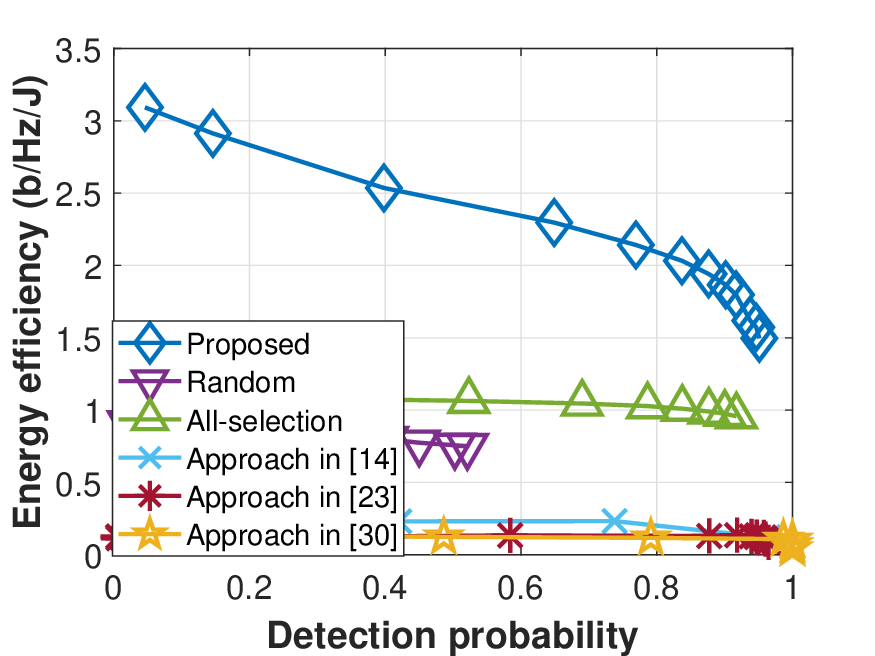}
        \caption{Number of RF chains is $8$.}
        \label{fig: FCsmall_8RF_tradeoff}
    \end{subfigure}
    \vspace{-10pt}
    \caption{EE-detection probability tradeoff for approaches with FC architecture in setup 1.}
    \vspace{-10pt}
    \label{fig: FCsmall_tradeoff}
\end{figure}

\begin{figure}[!t]
    \begin{subfigure}{0.5\textwidth}
    \centering
        \includegraphics[width=0.8\linewidth]{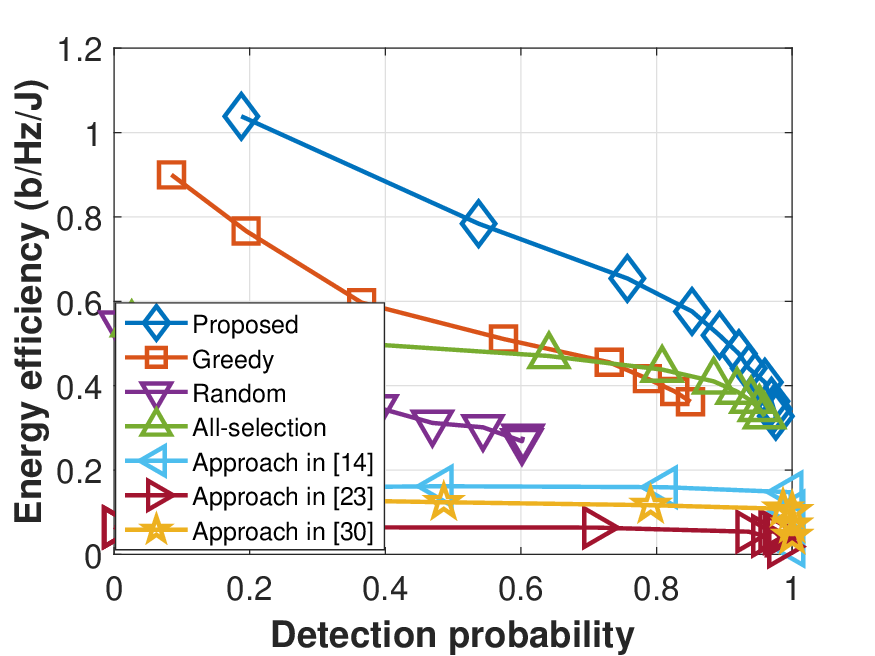} 
        \caption{Number of RF chains is $8$.}
        \label{fig: FClarge_8RF_tradeoff}
    \end{subfigure}
    \begin{subfigure}{0.5\textwidth}
    \centering
        \includegraphics[width=0.8\linewidth]{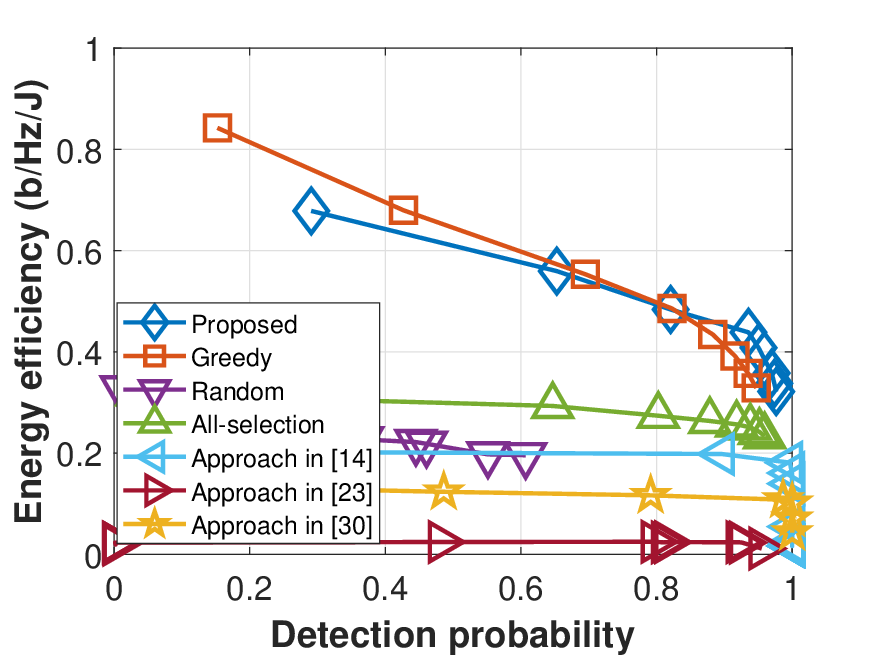}
        \caption{Number of RF chains is $16$.}
        \label{fig: FClarge_16RF_tradeoff}
    \end{subfigure}
    \vspace{-10pt}
    \caption{EE-detection probability tradeoff for approaches with FC architecture in setup 2.}
    \vspace{-5pt}
    \label{fig: FClarge_tradeoff}
\end{figure}


\begin{figure}[!t]
    \begin{subfigure}{0.5\textwidth}
    \centering
        \includegraphics[width=0.8\linewidth]{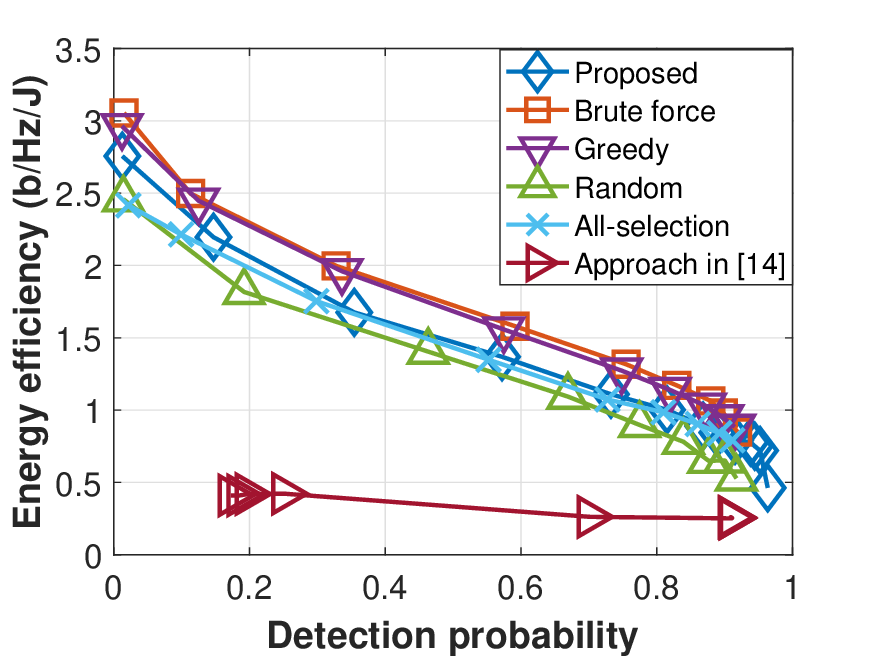} 
        \caption{Number of RF chains is $4$.}
        \label{fig: PCsmall_4RF_tradeoff}
    \end{subfigure}
    \begin{subfigure}{0.5\textwidth}
    \centering
        \includegraphics[width=0.8\linewidth]{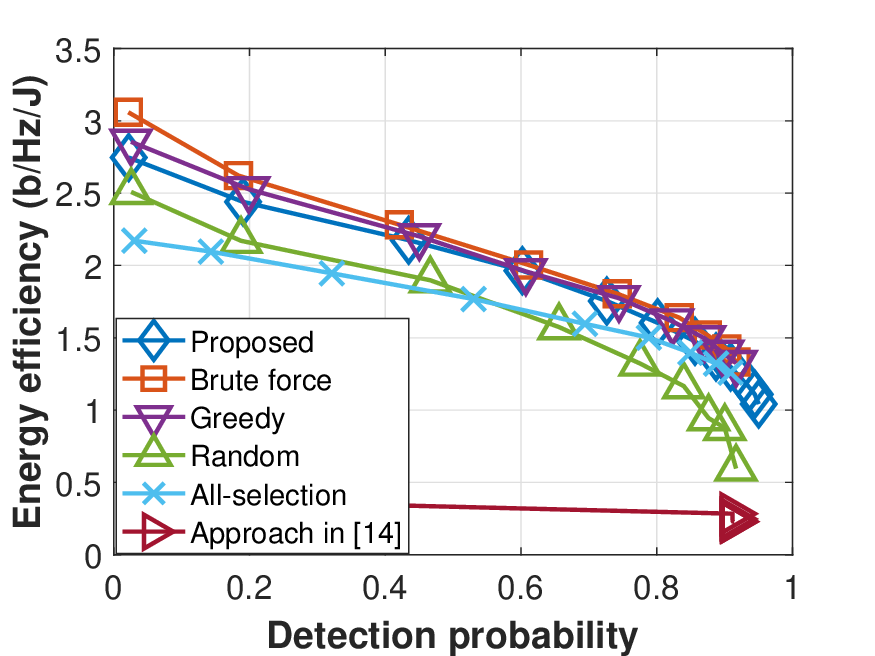}
        \caption{Number of RF chains is $8$.}
        \label{fig: PCsmall_8RF_tradeoff}
    \end{subfigure}
    \vspace{-10pt}
    \caption{EE-detection probability tradeoff for approaches with PC architecture in setup 1.}
    \label{fig: PCsmall_tradeoff}
\end{figure}


\begin{figure}[!t]
   \centering
\includegraphics[width=1.0\linewidth]{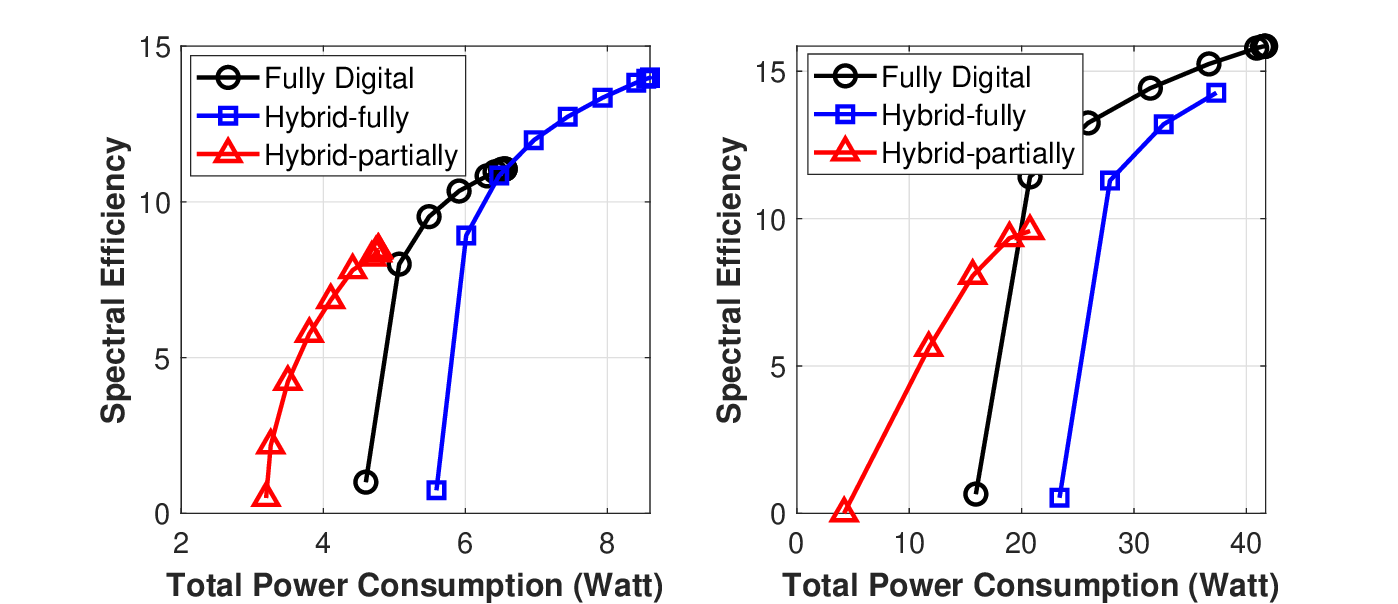}
   \vspace{-5pt}
   \caption{Evaluation of the tradeoff between spectral efficiency and power consumption. The left-hand figure is with setup 1 and the hybrid structures have $N_\text{RF}=4$. The left-hand figure is with setup 2 and the hybrid structures have $N_\text{RF}=8$.}
   \vspace{-5pt}
   \label{fig:SE_Power_Tradeoff}
\end{figure}

\begin{figure}[!t]
    \centering
\includegraphics[width=0.8\linewidth]{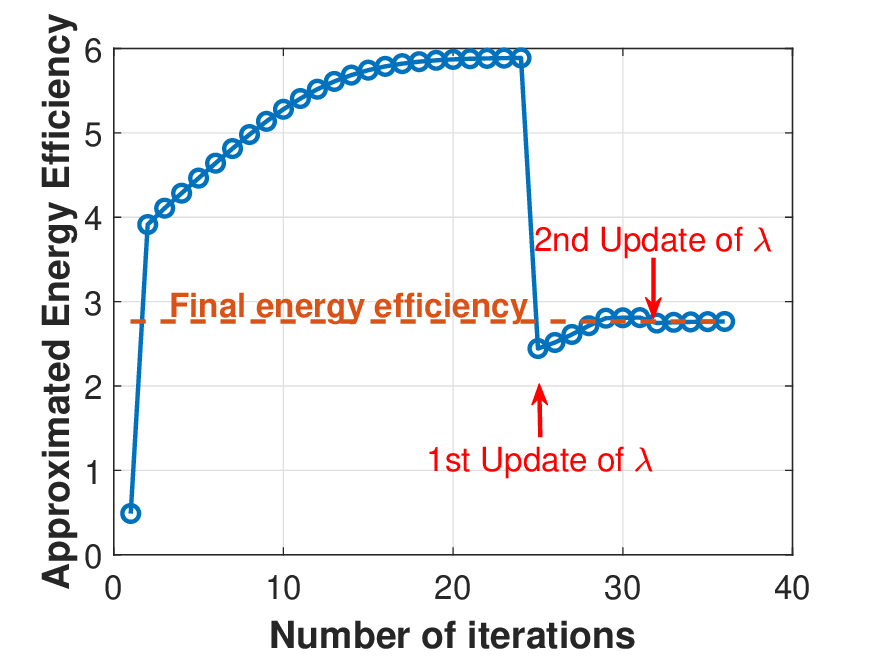}
    \caption{Convergence evaluation of Alg. \ref{alg_FDInnerLoop} under setup 1.}
    \vspace{-10pt}
    \label{fig:EE_convergence}
\end{figure}

In Fig. \ref{fig:SE_Power_Tradeoff}, we evaluate the SE and power consumption tradeoff design approach discussed in Sec. \ref{sec_3}-D. Note the solution points of the original EE-best solutions indeed are located on their corresponding SE-power tradeoff curves. We see that our approach can successfully provide tradeoff between SE and energy consumption. In addition, while not under a very fair comparison, we observe that the design under fully-connected hybrid architecture is more energy-consuming than that of the fully-digital architecture. This might be because that the total energy consumption due to the large amount of phase shifters used in the fully-connected architecture could be large, making the fully-connected architecture less energy efficient. This implies that when adopting the hybrid architecture for high EE and low energy consumption, the partially-connected structure is more preferable.

In Table \ref{tab:runtime_setup1}, we provide the runtime complexity evaluation of different approaches in terms of number of seconds for obtaining the precoding and RF chain selection design one time in setups 1 and 2. We see that the brute-force-based approach has much higher complexity as compared to other approaches. In addition, our proposed iterative approach has lower complexity than the proposed greedy-based approach. Furthermore, when comparing between results in setups 1 and 2, our proposed iterative approach has a better scalability than the greedy-based approach. Finally, to validate the convergence of our proposed approaches, we evaluate the convergence of Alg. \ref{alg_FDInnerLoop}, since most of our proposed design approaches depend on it. The results under FD architecture of setup 1 are illustrated in Fig. \ref{fig:EE_convergence}, where the function value shown in the figure follows the objective function in \eqref{pr: WMMSE}. It can be observed when given a $\lambda$, the function value can be monotonically improved until convergence, validating the convergence property of our algorithm. 

\section{Conclusions}

To fill the gap of energy efficient designs for ISAC systems, we investigated the joint precoding and RF chain selection design approaches for both FD and hybrid precoding architectures in multi-user MIMO-OFDM ISAC systems from EE aspect. To this end, we first proposed an EE maximization problem that guarantees sensing performance. To solve the problem for the FD architecture, QT transformation and WMMSE conversion were used along with the SCA method and hyperbolic tangent-based approximation. Subsequently, the hybrid precoding and RF chain selection designs for both FC and PC architectures were developed based on the proposed FD design. Complexity and convergence analysis of our approaches were provided, showing the theoretical convergence guarantee of our approaches. Also, we discuss how to extend our design approaches to derive SE-power consumption tradeoff designs. Simulation results validates that our proposed designs significantly outperform the reference schemes in terms of the EE-detection tradeoff.

\begin{table}[!t]
  \centering
  \caption{Runtime Comparisons in Setups 1 and 2}
  \label{tab:runtime_setup1}
  \renewcommand{\arraystretch}{1.15}  
  \begin{tabular}{|c|c|c|c|}
    \hline
    \multicolumn{1}{|c|}{\textbf{Setup 1}} &
    \textbf{Proposed} &
    \textbf{Greedy } &
    \textbf{Brute force} 
     \\ \hline
    FD & 3.67   & 6.69 & 160.44  \\ \hline
    FC & 30.62  & --     & --    \\ \hline
    PC & 13.25  & 25.87  & 662.42  \\ \hline
    \multicolumn{1}{|c|}{\textbf{Setup 2}} &
    \textbf{Proposed} &
    \textbf{Greedy} &
    \\ \hline
    FD & 21.65  & 152.45 &    \\ \hline
    FC & 583.23 & --     &  \\ \hline
    PC & 55.86  & 267.69 &   \\ \hline
  \end{tabular}
\end{table}


\appendices

\bibliographystyle{IEEEtran}
\bibliography{Bib_3C_2020_11_23}

\begin{thebibliography}{10}
\providecommand{\url}[1]{#1}
\csname url@samestyle\endcsname
\providecommand{\newblock}{\relax}
\providecommand{\bibinfo}[2]{#2}
\providecommand{\BIBentrySTDinterwordspacing}{\spaceskip=0pt\relax}
\providecommand{\BIBentryALTinterwordstretchfactor}{4}
\providecommand{\BIBentryALTinterwordspacing}{\spaceskip=\fontdimen2\font plus
\BIBentryALTinterwordstretchfactor\fontdimen3\font minus
  \fontdimen4\font\relax}
\providecommand{\BIBforeignlanguage}[2]{{%
\expandafter\ifx\csname l@#1\endcsname\relax
\typeout{** WARNING: IEEEtran.bst: No hyphenation pattern has been}%
\typeout{** loaded for the language `#1'. Using the pattern for}%
\typeout{** the default language instead.}%
\else
\language=\csname l@#1\endcsname
\fi
#2}}
\providecommand{\BIBdecl}{\relax}
\BIBdecl

\bibitem{ZhangJCS2022}
J.~A. Zhang, M.~L. Rahman, K.~Wu, X.~Huang, Y.~J. Guo, S.~Chen, and J.~Yuan,
  ``Enabling joint communication and radar sensing in mobile networks—a
  survey,'' \emph{IEEE Communications Surveys \& Tutorials}, vol.~24, no.~1,
  pp. 306--345, Firstquarter 2022.

\bibitem{Lu2024ISAC10}
S.~Lu, F.~Liu, Y.~Li, K.~Zhang, H.~Huang, J.~Zou, X.~Li, Y.~Dong, F.~Dong,
  J.~Zhu, Y.~Xiong, W.~Yuan, Y.~Cui, and L.~Hanzo, ``Integrated sensing and
  communications: Recent advances and ten open challenges,'' \emph{IEEE
  Internet of Things Journal}, vol.~11, no.~11, pp. 19\,094--19\,120, Jun.
  2024.

\bibitem{braun2014ofdm}
K.~M. Braun, ``{OFDM} radar algorithms in mobile communication networks,''
  Ph.D. dissertation, Karlsruhe, Karlsruher Institut f{\"u}r Technologie (KIT),
  Diss., 2014, 2014.

\bibitem{carvajal2020comparison}
G.~K. Carvajal, M.~F. Keskin, C.~Aydogdu, O.~Eriksson, H.~Herbertsson,
  H.~Hellsten, E.~Nilsson, M.~Rydstr{\"o}m, K.~V{\"a}nas, and H.~Wymeersch,
  ``Comparison of automotive {FMCW} and {{OFDM}} radar under interference,'' in
  \emph{Proc. IEEE Radar Conf.}, 2020, pp. 1--6.

\bibitem{dai2025tutorial}
Q.~Dai, Y.~Zeng, H.~Wang, C.~You, C.~Zhou, H.~Cheng, X.~Xu, S.~Jin, A.~L.
  Swindlehurst, Y.~C. Eldar \emph{et~al.}, ``A tutorial on {MIMO-OFDM ISAC}:
  From far-field to near-field,'' \emph{arXiv preprint arXiv:2504.19091}, 2025.

\bibitem{Zhang20196G}
Z.~Zhang, Y.~Xiao, Z.~Ma, M.~Xiao, Z.~Ding, X.~Lei, G.~K. Karagiannidis, and
  P.~Fan, ``6g wireless networks: Vision, requirements, architecture, and key
  technologies,'' \emph{IEEE Vehicular Technology Magazine}, vol.~14, no.~3,
  pp. 28--41, Sep. 2019.

\bibitem{Gao2016EE}
X.~Gao, L.~Dai, S.~Han, C.-L. I, and R.~W. Heath, ``Energy-efficient hybrid
  analog and digital precoding for mm{W}ave {MIMO} systems with large antenna
  arrays,'' \emph{IEEE Journal on Selected Areas in Communications}, vol.~34,
  no.~4, pp. 998--1009, Apr. 2016.

\bibitem{he2022energy}
Z.~He, W.~Xu, H.~Shen, Y.~Huang, and H.~Xiao, ``Energy efficient beamforming
  optimization for integrated sensing and communication,'' \emph{IEEE Wireless
  Commun. Lett.}, vol.~11, no.~7, pp. 1374--1378, Apr. 2022.

\bibitem{ZouEEISAC2022}
J.~Zou, Y.~Cui, Y.~Liu, and S.~Sun, ``Energy efficiency optimization for
  integrated sensing and communications systems,'' in \emph{2022 IEEE Wireless
  Communications and Networking Conference (WCNC)}, Apr. 2022, pp. 216--221.

\bibitem{zou2024energy}
J.~Zou, S.~Sun, C.~Masouros, Y.~Cui, Y.-F. Liu, and D.~W.~K. Ng,
  ``Energy-efficient beamforming design for integrated sensing and
  communications systems,'' \emph{IEEE Trans. Commun.}, Feb. 2024.

\bibitem{allu2024robust}
R.~Allu, M.~Katwe, K.~Singh, T.~Q. Duong, and C.-P. Li, ``Robust energy
  efficient beamforming design for {ISAC} full-duplex communication systems,''
  \emph{IEEE Wireless Commun. Lett.}, Jun. 2024.

\bibitem{singh2024energy}
J.~Singh, S.~Srivastava, and A.~K. Jagannatham, ``Energy-efficient hybrid
  beamforming for integrated sensing and communication enabled {mmWave} {MIMO}
  systems,'' in \emph{2024 IEEE 99th Vehicular Technology Conference
  (VTC2024-Spring)}, Dec. 2024, pp. 1--6.

\bibitem{kaushik2021hardware}
A.~Kaushik, C.~Masouros, and F.~Liu, ``Hardware efficient joint
  radar-communications with hybrid precoding and {RF} chain optimization,'' in
  \emph{Proc. IEEE Int. Conf. Commun.}, 2021, pp. 1--6.

\bibitem{kaushik2022green}
A.~Kaushik, E.~Vlachos, C.~Masouros, C.~Tsinos, and J.~Thompson, ``Green joint
  radar-communications: {RF} selection with low resolution {DAC}s and hybrid
  precoding,'' in \emph{Proc. IEEE Int. Conf. Commun.}, 2022, pp. 3160--3165.

\bibitem{dizdar2022energy}
O.~Dizdar, A.~Kaushik, B.~Clerckx, and C.~Masouros, ``Energy efficient
  dual-functional radar-communication: Rate-splitting multiple access,
  low-resolution {DAC}s, and {RF} chain selection,'' \emph{IEEE Open J. Commun.
  Soc.}, vol.~3, pp. 986--1006, Jun. 2022.

\bibitem{WuEEISAC2024}
G.~Wu, Y.~Fang, J.~Xu, Z.~Feng, and S.~Cui, ``Energy-efficient mimo integrated
  sensing and communications with on–off nontransmission power,'' \emph{IEEE
  Internet of Things Journal}, vol.~11, no.~7, pp. 12\,177--12\,191, 2024.

\bibitem{huang2023capacity}
Z.~Huang, A.~Liu, R.~Du, and T.~X. Han, ``Capacity-crb tradeoff in {OFDM}
  integrated sensing and communication systems,'' in \emph{Proc. IEEE Int.
  Conf. Commun.}, 2023, pp. 2437--2442.

\bibitem{huang2022designing}
Y.~Huang, S.~Hu, S.~Ma, Z.~Liu, and M.~Xiao, ``Designing low-{PAPR} waveform
  for {{OFDM}}-based radcom systems,'' \emph{IEEE Trans. Wireless Commun.},
  vol.~21, no.~9, pp. 6979--6993, Mar. 2022.

\bibitem{hsu2021analysis}
H.-W. Hsu, M.-C. Lee, M.-X. Gu, Y.-C. Lin, and T.-S. Lee, ``Analysis and design
  for pilot power allocation and placement in {{OFDM}} based integrated radar
  and communication in automobile systems,'' \emph{IEEE Trans. Veh. Technol.},
  vol.~71, no.~2, pp. 1519--1535, Nov. 2022.

\bibitem{xu2020joint}
Z.~Xu, A.~Petropulu, and S.~Sun, ``A joint design of {MIMO}-{OFDM}
  dual-function radar communication system using generalized spatial
  modulation,'' in \emph{2020 IEEE Radar Conference (RadarConf20)}.\hskip 1em
  plus 0.5em minus 0.4em\relax IEEE, 2020, pp. 1--6.

\bibitem{xiao2024novel}
Z.~Xiao, R.~Liu, M.~Li, Q.~Liu, and A.~L. Swindlehurst, ``A novel joint
  angle-range-velocity estimation method for {{MIMO}-{OFDM}} {{ISAC}}
  systems,'' \emph{IEEE Trans. Signal Process.}, Aug. 2024.

\bibitem{tian2021transmit}
T.~Tian, T.~Zhang, L.~Kong, and Y.~Deng, ``Transmit/receive beamforming for
  {MIMO}-{OFDM} based dual-function radar and communication,'' \emph{IEEE
  Trans. Veh. Technol.}, vol.~70, no.~5, pp. 4693--4708, 2021.

\bibitem{nguyen2023multiuser}
N.~T. Nguyen, N.~Shlezinger, Y.~C. Eldar, and M.~Juntti, ``Multiuser {MIMO}
  wideband joint communications and sensing system with subcarrier
  allocation,'' \emph{IEEE Trans. Signal Process.}, Aug. 2023.

\bibitem{wei2024precoding}
Z.~Wei, R.~Yao, X.~Yuan, H.~Wu, Q.~Zhang, and Z.~Feng, ``Precoding optimization
  for {MIMO}-{OFDM} integrated sensing and communication systems,'' \emph{IEEE
  Trans. Cogn. Commun. Netw.}, 2024.

\bibitem{Xiao2025SpaseISAC}
Z.~Xiao, R.~Liu, M.~Li, W.~Wang, and Q.~Liu, ``Sparsity exploitation via joint
  receive processing and transmit beamforming design for mimo-ofdm isac
  systems,'' \emph{IEEE Transactions on Communications}, vol.~73, no.~5, pp.
  3593--3607, May 2025.

\bibitem{Hu2022low}
X.~Hu, C.~Masouros, F.~Liu, and R.~Nissel, ``Low-{PAPR} {DFRC} {MIMO}-{OFDM}
  waveform design for integrated sensing and communications,'' in \emph{Proc.
  IEEE Int. Conf. Commun.}, 2022, pp. 1599--1604.

\bibitem{wei2023waveform}
Z.~Wei, J.~Piao, X.~Yuan, H.~Wu, J.~A. Zhang, Z.~Feng, L.~Wang, and P.~Zhang,
  ``Waveform design for {MIMO-OFDM} integrated sensing and communication
  system: An information theoretical approach,'' \emph{IEEE Trans. Commun.},
  Sept. 2023.

\bibitem{liao2024beamforming}
Z.-T. Liao, M.-C. Lee, S.-F. Wu, and T.-S. Lee, ``Beamforming design for
  multi-user {MIMO}-{OFDM} integrated sensing and communication systems,'' in
  \emph{Proc. IEEE Int. Conf. Commun.}, 2024, pp. 4488--4493.

\bibitem{chou2024robust}
C.-H. Chou, M.-C. Lee, P.-C. Kang, and T.-S. Lee, ``Robust beamforming design
  for {MIMO}-{OFDM} joint radar and communication systems,'' in \emph{Proc.
  IEEE Int. Conf. Commun.}, 2024, pp. 1189--1194.

\bibitem{liao2025design}
Z.-T. Liao, S.-F. Wu, M.-C. Lee, T.-C. Chiu, and T.-S. Lee, ``Design of joint
  transmit beamforming for multi-user mimo-ofdm integrated sensing and
  communication systems,'' \emph{IEEE Transactions on Wireless Communications},
  2025.

\bibitem{liu2021dual}
R.~Liu, M.~Li, Q.~Liu, and A.~L. Swindlehurst, ``Dual-functional
  radar-communication waveform design: A symbol-level precoding approach,''
  \emph{IEEE J. Sel. Top. Signal Process.}, vol.~15, no.~6, pp. 1316--1331,
  2021.

\bibitem{Li2025ISACSidelobe}
P.~Li, M.~Li, R.~Liu, Q.~Liu, and A.~Lee~Swindlehurst, ``{MIMO-OFDM} {ISAC}
  waveform design for range-doppler sidelobe suppression,'' \emph{IEEE
  Transactions on Wireless Communications}, vol.~24, no.~2, pp. 1001--1015,
  Feb. 2025.

\bibitem{Wu2025LowISAC}
S.-F. Wu, M.-C. Lee, and T.-S. Lee, ``Symbol-level precoding-based waveform
  design for low-{PAPR} multi-user {MIMO-OFDM} integrated sensing and
  communication systems,'' in \emph{2025 IEEE International Conference on
  Communications (ICC)}, Jun. 2025.

\bibitem{liu2024joint}
Y.-S. Liu, S.-F. Wu, M.-C. Lee, C.-H. Chou, and T.-S. Lee, ``Joint beamforming
  and subcarrier allocation design for {MIMO-OFDM} dual-functional radar and
  communication systems,'' in \emph{2024 IEEE 100th Veh. Technol. Conf.
  (VTC2024-Fall)}.\hskip 1em plus 0.5em minus 0.4em\relax IEEE, Dec. 2024, pp.
  1--7.

\bibitem{Li2025Sparse}
X.~Li, H.~Min, Y.~Zeng, S.~Jin, L.~Dai, Y.~Yuan, and R.~Zhang, ``Sparse {MIMO}
  for {ISAC}: New opportunities and challenges,'' \emph{IEEE Wireless
  Communications}, pp. 1--9, 2025.

\bibitem{Luo2025ISACStandard}
X.~Luo, Q.~Lin, R.~Zhang, H.-H. Chen, X.~Wang, and M.~Huang, ``{ISAC} – a
  survey on its layered architecture, technologies, standardizations,
  prototypes and testbeds,'' \emph{IEEE Communications Surveys \& Tutorials},
  pp. 1--1, 2025.

\bibitem{li2020dynamic}
H.~Li, M.~Li, Q.~Liu, and A.~L. Swindlehurst, ``Dynamic hybrid beamforming with
  low-resolution pss for wideband {mmWave} {MIMO}-{OFDM} systems,'' \emph{IEEE
  J. Sel. Areas Commun.}, vol.~38, no.~9, pp. 2168--2181, Jun. 2020.

\bibitem{shen2018fractional}
K.~Shen and W.~Yu, ``Fractional programming for communication systems—part
  {I}: Power control and beamforming,'' \emph{IEEE Trans. Signal Process.},
  vol.~66, no.~10, pp. 2616--2630, May 2018.

\bibitem{SCA}
M.~Razaviyayn, ``Successive convex approximation: Analysis and applications,''
  Ph.D. dissertation, University of Minnesota, 2014.

\bibitem{yu2016alternating}
X.~Yu, J.-C. Shen, J.~Zhang, and K.~B. Letaief, ``Alternating minimization
  algorithms for hybrid precoding in millimeter wave {MIMO} systems,''
  \emph{IEEE J. Sel. Topics Signal Process.}, vol.~10, no.~3, pp. 485--500,
  2016.

\bibitem{BCD}
P.~Tseng, ``Convergence of a block coordinate descent method for
  nondifferentiable minimization,'' \emph{J. Optim. theory Appl.}, vol. 109,
  pp. 475--494, 2001.

\bibitem{wright1997primal}
S.~J. Wright, \emph{Primal-dual interior-point methods}.\hskip 1em plus 0.5em
  minus 0.4em\relax SIAM, 1997.

\bibitem{samimi20163}
M.~K. Samimi and T.~S. Rappaport, ``3-{D} millimeter-wave statistical channel
  model for {5G} wireless system design,'' \emph{IEEE Trans. Microw. Theory
  Techn.}, vol.~64, no.~7, pp. 2207--2225, Jun. 2016.

\end{thebibliography}
\end{document}